\documentclass{article}
\usepackage{amssymb}
\usepackage{amsfonts}
\usepackage{amsmath}
\usepackage{epsfig,euscript,enumerate,cancel}
\usepackage[enableskew]{youngtab}
\usepackage[all]{xy}
\usepackage{graphicx}

\newtheorem{theorem}{Theorem}[section]

\newtheorem{corollary}[theorem]{Corollary}

\newtheorem{definition}[theorem]{Definition}
\newtheorem{example}[theorem]{Example}

\newtheorem{proposition}[theorem]{Proposition}
\newtheorem{remark}[theorem]{Remark}

\newenvironment{proof}[1][Proof]{\noindent\textbf{#1.} }{\ \rule{0.5em}{0.5em}}
\numberwithin{equation}{section}
\numberwithin{figure}{section}
\input{tcilatex}

\begin{document}

\title{Noncommutative Schur polynomials and the crystal limit of the $U_{q}%
\widehat{\mathfrak{sl}}(2)$-vertex model}
\author{Christian Korff \\
%EndAName
School of Mathematics and Statistics, University of Glasgow\\
15 University Gardens, Glasgow G12 8QW, Scotland, UK\\
E-mail: c.korff@maths.gla.ac.uk
}
\maketitle

\begin{abstract}
Starting from the Verma module of $U_{q}\mathfrak{sl}(2)$ we consider the
evaluation module for affine $U_{q}\widehat{\mathfrak{sl}}(2)$ and discuss
its crystal limit ($q\rightarrow 0$). There exists an associated integrable
statistical mechanics model on a square lattice defined in terms of vertex
configurations. Its transfer matrix is the generating function for
noncommutative \emph{complete} symmetric polynomials in the generators of
the affine plactic algebra, an extension of the finite plactic algebra first
discussed by Lascoux and Sch\"{u}tzenberger. The corresponding
noncommutative \emph{elementary} symmetric polynomials were recently shown
to be generated by the transfer matrix of the so-called phase model discussed by Bogoliubov, Izergin and Kitanine. Here we establish that both generating
functions satisfy Baxter's TQ-equation in the crystal limit by tying them to
special $U_{q}\widehat{\mathfrak{sl}}(2)$ solutions of the Yang-Baxter
equation. The TQ-equation amounts to the well-known Jacobi-Trudi formula
leading naturally to the definition of noncommutative Schur polynomials. The
latter can be employed to define a ring which has applications in conformal
field theory and enumerative geometry: it is isomorphic to the fusion ring
of the $\widehat{\mathfrak{sl}}(n)_{k}$ -WZNW model whose structure
constants are the dimensions of spaces of generalized $\theta $-functions
over the Riemann sphere with three punctures.\bigskip\\
PACS numbers: 02.30.Ik, 05.50.+q, 11.25.Hf, 02.10.Hh, 02.10.Ox
\end{abstract}

\section{Introduction}

Integrable systems have many connections with different areas in pure
mathematics. In this article we shall focus on combinatorial aspects of a
particular quantum integrable system, the exactly solvable, statistical
vertex model associated with the quantum affine algebra $U_{q}\widehat{%
\mathfrak{sl}}(2)$ and \textquotedblleft infinite\textquotedblright\ spin.
For spin 1/2 this model specialises to the the well-known six-vertex model
or XXZ quantum spin-chain. By taking the crystal limit \cite{Ka} ($%
q\rightarrow 0$) one arrives at a drastically simplified version of this
model exhibiting nice combinatorial features: the Takahashi-Satsuma cellular
automaton \cite{TS} (or box and ball system); see also e.g. \cite{FOY}, \cite%
{HHIKTT} and references therein for the case of higher (finite) spin and
rank.\smallskip

In the case of infinite spin, i.e. each site of the chain now carries an
infinite-dimensional representation of $U_{q}\widehat{\mathfrak{sl}}(2)$
instead of a finite-dimensional one, there exists a link with enumerative
geometry: the commuting transfer matrices generate a ring of symmetric
polynomials in a noncommutative alphabet, the generators of the \emph{affine}
plactic algebra, whose \emph{finite} version has been introduced by Lascoux
and Sch\"{u}tzenberger \ \cite{LasSchutz} (see also \cite{FG} for a
discussion of the \emph{finite} plactic algebra in the context of
noncommutative Schur polynomials). It was shown in \cite[Part I]{KS} that
the noncommutative Schur polynomials related to the affine plactic algebra
can be employed to define a ring which is isomorphic to the fusion ring of
the $\widehat{\mathfrak{sl}}(n)_{k}$ Wess-Zumino-Novikov-Witten (WZNW)
model, a conformal field theory (CFT) with particularly nice algebraic and
geometric aspects. Here $k\geq 0$ is a non-negative integer called the
\textquotedblleft level\textquotedblright . The fusion ring is one of the
essential data of a CFT \cite{CFTbook} and in the case of the WZNW model its
structure constants coincide with the dimensions of spaces of generalized $%
\theta $-functions over the Riemann sphere with three punctures; see e.g.
\cite{Beauville}.\smallskip

In this article we shall discuss how the combinatorial description of the
fusion ring presented in \cite[Part I]{KS} is obtained from the crystal
limit of the $U_{q}\widehat{\mathfrak{sl}}(2)$ model with infinite spin. In
Section 2 we demonstrate on a simple example how under the action of the
affine plactic algebra the state space of the model decomposes into nice
lattices, called \textquotedblleft crystals\textquotedblright\, which depend
on the level $k$ and can be described in terms of coloured directed graphs
(usually called crystal graphs) related to the quantum affine algebra $U_{q}%
\widehat{\mathfrak{sl}}(n)$ with $n>2$ being the number of lattice sites.
Section 3 contains the definition of the $U_{q}\widehat{\mathfrak{sl}}(2)$
vertex model with infinite spin in the crystal limit. In particular we show
in Section 4 that the generating function of the \emph{noncommutative
complete symmetric polynomials} coincides with its transfer matrix; see
Proposition \ref{main} in the text. This new result complements the
discussion in \cite[Part I, Section 4,5]{KS} where the transfer matrix of
the phase model, first introduced by Bogoliubov, Izergin and Kitanine \cite{Bogoliubovetal},
has been identified with the generating function of the \emph{noncommutative}
\emph{elementary symmetric polynomials} \cite[Prop 5.13]{KS}. As in the ring
of symmetric functions over commutative variables, also in the present case
both sets of polynomials in noncommutative variables are linked via a determinant formula (a special case of the Jacobi-Trudi identity) for which we shall state an alternative proof to the one given in \cite[Def 5.10 and Cor 6.9]{KS}.
Namely, employing the relation to the $U_{q}\widehat{\mathfrak{sl}}(2)$
algebra we show that the transfer matrix of the $U_{q}\widehat{\mathfrak{sl}}%
(2)$ model with infinite spin and the transfer matrix of the phase model
obey Baxter's famous TQ-equation \cite{Baxter} in the crystal limit. In the
concluding section we summarise the connection with the WZNW fusion ring,
giving a brief account of some of the main results from \cite[Part I]{KS},
and state a novel recursion formula for fusion coefficients.

\section{The crystal limit of the $U_{q}\widehat{\mathfrak{sl}}(2)$-Verma
module}

\label{verma}

We set out by introducing the central algebraic structure: we recall the
definition of the quantum affine algebra $U_{q}\widehat{\mathfrak{sl}}(2)$
and then consider a particular infinite-dimensional representation of it.

\begin{definition}
The q-deformed universal enveloping algebra $U_{q}\widehat{\mathfrak{sl}}(2)$
is the unital associative algebra over $\mathbb{C}(q)$ generated from the
letters $\{E_{i},F_{i},K_{i}^{\pm 1}\}_{i=0,1}$ subject to the algebraic
relations
\begin{gather}
K_{i}E_{j}K_{i}^{-1}=q^{A_{ij}}E_{j},\quad
K_{i}F_{j}K_{i}^{-1}=q^{-A_{ij}}F_{j}\;,  \label{AQG} \\
\lbrack E_{i},F_{j}]=\delta _{ij}\frac{K_{i}-K_{i}^{-1}}{q-q^{-1}},\quad
K_{i}K_{j}=K_{j}K_{i},  \notag
\end{gather}%
and for $i\neq j$
\begin{equation}
\sum_{p=0}^{1-A_{ij}}(-1)^{p}\QATOPD[ ] {1-A_{ij}}{p}%
_{q}X_{i}^{1-A_{ij}-p}X_{j}X_{i}^{p}=0,\qquad X_{i}=E_{i},F_{i},  \label{CS}
\end{equation}%
where $A=\left(
\begin{smallmatrix}
2 & -2 \\
-2 & 2%
\end{smallmatrix}%
\right) $ is the Cartan matrix\ of $\widehat{\mathfrak{sl}}(2)$ and we have
set $[x]_{q}=(q^{x}-q^{-x})/(q-q^{-1}),$ $\QATOPD[ ] {m}{n}_{q}:=\frac{%
[m]_{q}!}{\left[ n\right] _{q}!\left[ m-n\right] _{q}!}$, $%
[n]_{q}!:=\prod_{k=1}^{n}[k]_{q}$ as usual. We denote by $U_{q}\mathfrak{sl}%
(2)\subset U_{q}\widehat{\mathfrak{sl}}(2)$ the subalgebra generated from $%
\{E_{1},F_{1},K_{1}^{\pm 1}\}$.
\end{definition}

We introduce the following coproduct $\Delta :U_{q}\widehat{\mathfrak{sl}}%
(2)\rightarrow U_{q}\widehat{\mathfrak{sl}}(2)\otimes U_{q}\widehat{%
\mathfrak{sl}}(2)$,
\begin{equation}
\Delta (E_{i})=E_{i}\otimes 1+K_{i}\otimes E_{i},\;\Delta
(F_{i})=F_{i}\otimes K_{i}^{-1}+1\otimes F_{i},\;\Delta (K_{i})=K_{i}\otimes
K_{i}~.  \label{cop}
\end{equation}%
It is well-known that $U_{q}\widehat{\mathfrak{sl}}(2)$ can be turned into a
Hopf algebra by defining in addition a counit and antipode; see e.g. \cite%
{CP, HK}. However, as we shall not use the latter maps here, we omit their
definition. We now consider a particular module of $U_{q}\widehat{\mathfrak{%
sl}}(2)$ which we shall use throughout this article, first to define various
other related algebras in the crystal limit and then to define certain
statistical mechanics models.

First we recall the definition of a Verma module for the finite subalgebra $%
U_{q}\mathfrak{sl}(2)$. Let $\mathcal{I}_{\mu }\subset U_{q}\mathfrak{sl}(2)$
be the left ideal generated by $E_{1}$ and $K_{1}-\mu q^{-1}1$. Consider the
module $M_{\mu }=U_{q}\mathfrak{sl}(2)/\mathcal{I}_{\mu }$ which is free
over the subalgebra generated by $F_{1}$ (by a quantum version of the Poincar%
\'{e}-Birkhoff-Witt theorem). Set $v_{0}=1+\mathcal{I}_{\mu }$ and $%
v_{m}:=F_{1}^{m}v_{0}$ for all $m\in \mathbb{N}$, then $\{v_{m}\}_{m\in
\mathbb{Z}_{\geq 0}}$ is a basis and the following relations hold%
\begin{eqnarray}
K_{1}v_{m} &=&\mu q^{-2m-1}v_{m},\qquad F_{1}v_{m}=v_{m+1},  \label{QGaction}
\\
E_{1}^{m}v_{0} &=&0,\qquad E_{1}v_{m}=\left( \tfrac{\mu q^{-m}-\mu ^{-1}q^{m}%
}{q-q^{-1}}\right) [m]_{q}v_{m-1}\ .  \notag
\end{eqnarray}%
Consider $M_{\mu }(u)=\mathbb{C}(q)[u,u^{-1}]\otimes _{\mathbb{C}(q)}M_{\mu
} $ then we can regard $M_{\mu }(u)$ as a $U_{q}\widehat{\mathfrak{sl}}(2)$%
-module by defining the following action of the affine generators
\begin{equation}
E_{0}\rightarrow u^{-1}\otimes F_{1},\qquad F_{0}\rightarrow u\otimes
E_{1},\qquad K_{0}^{\pm 1}\rightarrow 1\otimes K_{1}^{\mp 1}
\end{equation}%
while keeping the action of the non-affine generators unchanged, $%
X_{1}\rightarrow 1\otimes X_{1}$ with $X_{1}=E_{1},F_{1},K_{1}^{\pm 1}$. For
$\upsilon \in \mathbb{C}^{\times }$ the corresponding evaluation module is
the one obtained by taking the quotient in $\mathbb{C}(q)[u,u^{-1}]$ with
respect to the maximal ideal generated by $u-\upsilon $.

\begin{remark}
\label{reduction}\textrm{For generic values of $\mu $ this module is known
to be irreducible. If we choose $\mu =q^{d}$ for some positive integer $d$,
we have $E_{1}^{d}v_{d}=0$ and we can identify the submodule spanned by the
first $d$ vectors $v_{0},v_{1},\ldots ,v_{d-1}$ with the standard $U_{q}%
\mathfrak{sl}(2)$ module of dimension $d$ upon imposing the condition $%
F_{1}^{m}=0$ for $m\geq d$. In particular we recover the fundamental
representation for $d=2$.}
\end{remark}

\subsection{Crystal limit and phase maps}

\label{crystalverma}

We have defined the quantum algebra $U_{q}\widehat{\mathfrak{sl}}(2)$ over
the ring of rational functions in the indeterminate $q$. Naively speaking we
now wish to take the limit $q\rightarrow 0$, which is referred to as the
\textquotedblleft crystal limit\textquotedblright\ as it was originally
considered in the context of statistical mechanics models, where it
corresponds to the low temperature limit when the system \textquotedblleft
crystallizes\textquotedblright\ into a single configuration. In more
technical terms we restrict the algebra into what follows to the ring of
regular functions in $q$ and then take the quotient with respect to ideal
generated by $q$. We shall refer to this procedure as the crystal limit in
accordance with the literature; see for example \cite[Section 4.2, page 67]%
{HK} and references therein. Our main interest are the combinatorial
features which emerge in this limit.

\begin{proposition}
Denote by $\mathbb{A}\subset \mathbb{C}(q)$ the ring of functions which are
regular at $q=0$ and let $\mathcal{L}=\tbigoplus_{m\in \mathbb{Z}_{\geq 0}}%
\mathbb{A}v_{m}$. Then

\begin{itemize}
\item[(i)] $\mathcal{L}$ generates $M_{\mu }$ as a vector space over $%
\mathbb{C}(q),\;M_{\mu }=\mathbb{C}(q)\otimes _{\mathbb{A}}\mathcal{L}$.

\item[(ii)] Define $\tilde{E}_{1}=-q^{-1}K_{1}^{-1}E_{1}$ then $\tilde{E}_{1}%
\mathcal{L}\subset \mathcal{L}$ and $F_{1}\mathcal{L}\subset \mathcal{L}$.
\end{itemize}
\end{proposition}

\begin{proof}
The first assertion is obvious. The statements under (ii) are easily
verified from (\ref{QGaction}). For instance, we find that%
\begin{equation}
\tilde{E}_{1}v_{m}=\left[ \left( \frac{1-\mu ^{-2}q^{2m}}{1-q^{2}}\right)
\left( \frac{q^{2m}-1}{q^{2}-1}\right) \right] v_{m-1}  \label{crystalE}
\end{equation}%
where the coefficient is obviously regular at $q=0$. Similarly, we compute
the relation
\begin{equation}
\left[ \tilde{E}_{1}F_{1}-q^{2}F_{1}\tilde{E}_{1}\right] v_{m}=\left[ \frac{%
1-\mu ^{-2}q^{4m+2}}{1-q^{2}}\right] v_{m}  \label{crystalcomm}
\end{equation}%
which we will use below.
\end{proof}

In what follows we consider the crystal limit of $M_{\mu }$, that is we will
consider $\mathcal{M}=\mathcal{L}/q\mathcal{L}$ as $\mathbb{C}$-vector space
in the natural way. From (\ref{QGaction}) and (\ref{crystalE}) we infer that
the quantum algebra elements $K_{1},$ $\tilde{E}_{1}$ and $F_{1}$ then
induce the following maps $N,\varphi ,\varphi ^{\ast }:\mathcal{M}%
\rightarrow \mathcal{M}$ defined via%
\begin{equation}
Nv_{m}=mv_{m},\text{\quad }\varphi ^{\ast }v_{m}=v_{m+1}\quad \text{and\quad
}\varphi v_{m}:=\left\{
\begin{array}{cc}
0, & m=0 \\
v_{m-1}, & m>0%
\end{array}%
\right. \;.  \label{crystalNEF}
\end{equation}%
In particular, we have the following crystal limit of the $U_{q}\mathfrak{sl}%
(2)$ relations (see equations (\ref{crystalcomm})),%
\begin{equation}
\varphi \varphi ^{\ast }=1\qquad \text{and\qquad }\varphi ^{\ast }\varphi
v_{m}=\left\{
\begin{array}{cc}
0, & m=0 \\
v_{m}, & m>0%
\end{array}%
\right. \;.  \label{crystalQGrel}
\end{equation}%
Fix an integer $n>2$ and consider the tensor product $\mathcal{H}=\mathcal{M}%
^{\otimes n}$. We extend the maps (\ref{crystalNEF}) to $\mathcal{H}=%
\mathcal{M}^{\otimes n}$ by defining for $i=1,\ldots ,n$%
\begin{equation}
\varphi _{i}:=1\otimes \cdots \otimes \underset{i}{\varphi }\otimes \cdots
\otimes 1  \label{phase}
\end{equation}%
and similarly, $\varphi _{i}^{\ast }:=1\otimes \cdots \otimes \underset{i}{%
\varphi ^{\ast }}\otimes \cdots \otimes 1,\;N_{i}:=1\otimes \cdots \otimes
\underset{i}{N}\otimes \cdots \otimes 1$. We obtain the phase algebra discussed by Bogoliubov, Izergin and Kitanine in \cite{Bogoliubovetal}; see also references [5], [6] and [8] therein. The proof of the following result can be found in \cite[Prop 3.1]{KS}.

\begin{proposition}
The $\varphi _{i},\varphi _{i}^{\ast }$ and $N_{i}$ generate a subalgebra $%
\hat{\Phi}$ of $\func{End}({\mathcal{H}})$ which can be realized as the
algebra $\Phi $ with the following generators and relations for $1\leq
i,j\leq n$:
\begin{gather}
\varphi _{i}\varphi _{j}=\varphi _{j}\varphi _{i},\quad \varphi _{i}^{\ast
}\varphi _{j}^{\ast }=\varphi _{j}^{\ast }\varphi _{i}^{\ast },\quad
N_{i}N_{j}=N_{j}N_{i}  \label{comm} \\
N_{i}\varphi _{j}-\varphi _{j}N_{i}=-\delta _{ij}\varphi _{i},\quad
N_{i}\varphi _{j}^{\ast }-\varphi _{j}^{\ast }N_{i}=\delta _{ij}\varphi
_{i}^{\ast },  \label{comm2} \\
\varphi _{i}\varphi _{i}^{\ast }=1,\quad \varphi _{i}\varphi _{j}^{\ast
}=\varphi _{j}^{\ast }\varphi _{i}\;\text{ if }\;i\neq j,  \label{comm3} \\
N_{i}(1-\varphi _{i}^{\ast }\varphi _{i})=0=(1-\varphi _{i}^{\ast }\varphi
_{i})N_{i}.  \label{Npi}
\end{gather}%
If we introduce the scalar product on the vector space ${\mathcal{H}}$ by
\begin{equation*}
\langle \alpha v_{m_{1}}\otimes \cdots \otimes v_{m_{n}},\beta
v_{m_{1}^{\prime }}\otimes \cdots \otimes v_{m_{n}^{\prime }}\rangle =%
\overline{\alpha }\beta \tprod_{i=0}^{n-1}\delta _{m_{i},m_{i}^{\prime }},
\end{equation*}%
for $\alpha ,\beta \in \mathbb{C}$, then $\langle \varphi _{i}^{\ast
}v,v^{\prime }\rangle =\langle v,\varphi _{i}v^{\prime }\rangle \;$\ for any
$v,v^{\prime }\in {\mathcal{H}}$.
\end{proposition}

\subsection{Crystallisation of the state space}

\label{plactic}

We now decompose the tensor product $\mathcal{H}=\mathcal{M}^{\otimes n}$
into an infinite direct sum,%
\begin{equation}
\mathcal{H}=\tbigoplus_{k\in \mathbb{Z}_{\geq 0}}\mathcal{H}_{k},\qquad
\mathcal{H}_{k}=\mathbb{C}\left\{ v_{m_{1}}\otimes \cdots \otimes
v_{m_{n}}~\left\vert ~\tsum_{i=1}^{n}m_{i}=k\right. \right\} ,
\label{decomp}
\end{equation}%
where we set $\mathcal{H}_{0}=\mathbb{C}\{v_{0}\otimes \cdots \otimes
v_{0}\}\cong \mathbb{C}$ and the summation index $k$ is the
\textquotedblleft level\textquotedblright\ of the WZNW model. For notational
convenience we will often identify a basis vector $v_{m_{1}}\otimes \cdots
\otimes v_{m_{n}}$ in $\mathcal{H}$ with the composition $\boldsymbol{m}%
=(m_{1},\ldots ,m_{n})$ or, equivalently, the partition $\hat{\lambda}=(\hat{%
\lambda}_{1},\ldots ,\hat{\lambda}_{n})$ whose associated Young diagram
contains $m_{i}$ columns of height $i$. We denote the corresponding set of
such partitions by $P^{+}$ and the subset corresponding to $\mathcal{H}_{k}$
by $P_{k}^{+}$. Obviously each $\hat{\lambda}\in P_{k}^{+}$ has at most $n$
parts and $\hat{\lambda}_{1}=k$. We now wish to explain how the subspaces $%
\mathcal{H}_{k}$ can be identified with crystal graphs of the affine quantum
algebra $U_{q}\widehat{\mathfrak{sl}}(n)$.

\begin{definition}
Let $\mathcal{A}=\{a_{0},a_{1},a_{2},\ldots a_{n-1}\}$. The\ \emph{local\
affine plactic algebra} $\func{Pl}=\func{Pl}(\mathcal{A})$ is the free
algebra generated by the elements of $\mathcal{A}$ modulo the relations
\begin{eqnarray}
a_{i}a_{j}-a_{j}a_{i}=0, &&\text{ if $|i-j|\neq 1\mod n$},  \label{PL1} \\
a_{i+1}a_{i}^{2}=a_{i}a_{i+1}a_{i}, &&a_{i+1}^{2}a_{i}=a_{i+1}a_{i}a_{i+1},
\label{PL2}
\end{eqnarray}%
where in \eqref{PL2} all variables are understood as elements in $\mathcal{A}
$ by taking indices modulo n. Let $\func{Pl}_{\func{fin}}=\func{Pl}_{\func{%
fin}}(\mathcal{A}^{\prime })$ denote the \emph{local finite plactic algebra}
generated from $\mathcal{A}^{\prime }=\{a_{1},a_{2},\ldots a_{n-1}\}$;
compare with \cite{FG}.
\end{definition}

We recall the following result from \cite[Prop 5.8]{KS}:

\begin{proposition}
\label{faithfulness} There is a homomorphism of algebras $\func{Pl}_{\func{%
fin}}\rightarrow \Phi $ such that
\begin{equation}
a_{j}\mapsto \varphi _{j+1}^{\ast }\varphi _{j},\quad j=1,...,n-1,
\label{placticrep}
\end{equation}%
hence, the representation of the phase algebra $\Phi $ given by (\ref%
{crystalNEF}) and (\ref{phase}) lifts to a representation of the local
plactic algebra $\func{Pl}_{\func{fin}}$. Mapping $a_{0}=a_{n}$ to $z\varphi
_{1}^{\ast }\varphi _{n}$ it lifts in addition to a representation of $\func{%
Pl}$ on ${\mathcal{H}}[z]=\mathbb{C}(z)\otimes _{\mathbb{C}}{\mathcal{H}}$
with $z$ an indeterminate. Both representations are faithful.
\end{proposition}

As mentioned in \cite[Remark 5.9]{KS} the subspace $\mathcal{H}_{k}\subset
\mathcal{H}=\mathcal{M}^{\otimes n}$ together with the action (\ref%
{placticrep}) of the local affine plactic algebra can be identified with the
crystal graph of the $k^{\text{th}}$-symmetric tensor representation of the
vector representation of $U_{q}\widehat{\mathfrak{sl}}(n)$ \cite{JMMO}. In the literature this crystal graph is also known as the affinization of the Kirillov-Reshetikhin crystal graph $\mathfrak{B}_{1,k}$ of type $A$; see e.g. \cite[Section 3]{Okado}. We discuss an
explicit example below.

\begin{definition}
The quantum universal enveloping algebra $U_{q}\widehat{\mathfrak{sl}}(n),$ $%
n>2$ is the associative unital $\mathbb{C}(q)$-algebra generated by $%
\{E_{i},F_{i},K_{i}^{\pm 1}\}_{i=0}^{n-1}$ subject to the analogous
identities as in (\ref{AQG}), (\ref{CS}) but with respect to the $\widehat{%
\mathfrak{sl}}(n)$ Cartan matrix,
\begin{equation*}
A=\left(
\begin{smallmatrix}
2 & -1 & 0 & \cdots & 0 & -1 \\
-1 & 2 & \ddots &  &  & 0 \\
0 & \ddots & \ddots &  &  & \vdots \\
\vdots &  &  &  &  & 0 \\
0 &  &  &  & 2 & -1 \\
-1 & 0 & \cdots & 0 & -1 & 2%
\end{smallmatrix}%
\right) \ .
\end{equation*}
\end{definition}

We recall that the vector representation $V=\mathbb{C}\{v_{1},...,v_{n}\}$
associated with the fundamental weight $\omega _{1}$ is given by%
\begin{equation}
E_{i}v_{r}=\delta _{i,r-1}v_{r-1},\qquad F_{i}v_{r}=\delta
_{r,i}v_{r+1},\qquad K_{i}v_{r}=q^{\delta _{i,r}-\delta _{i,r-1}}v_{r}\;.
\end{equation}%
As in the case of $n=2$ this $U_{q}\mathfrak{sl}(n)$-module can be turned
into an evaluation module $V(z)$ for any $z\in \mathbb{C}^{\times }$ by
setting%
\begin{equation}
E_{0}v_{r}=z^{-1}\delta _{r,1}v_{n},\qquad F_{0}v_{r}=z~\delta
_{r,n}v_{1},\qquad K_{0}v_{r}=q^{\delta _{r,n}-\delta _{r,1}}v_{r}\;.
\end{equation}

\begin{figure}[tbp]
\begin{equation*}
\includegraphics[scale=0.35]{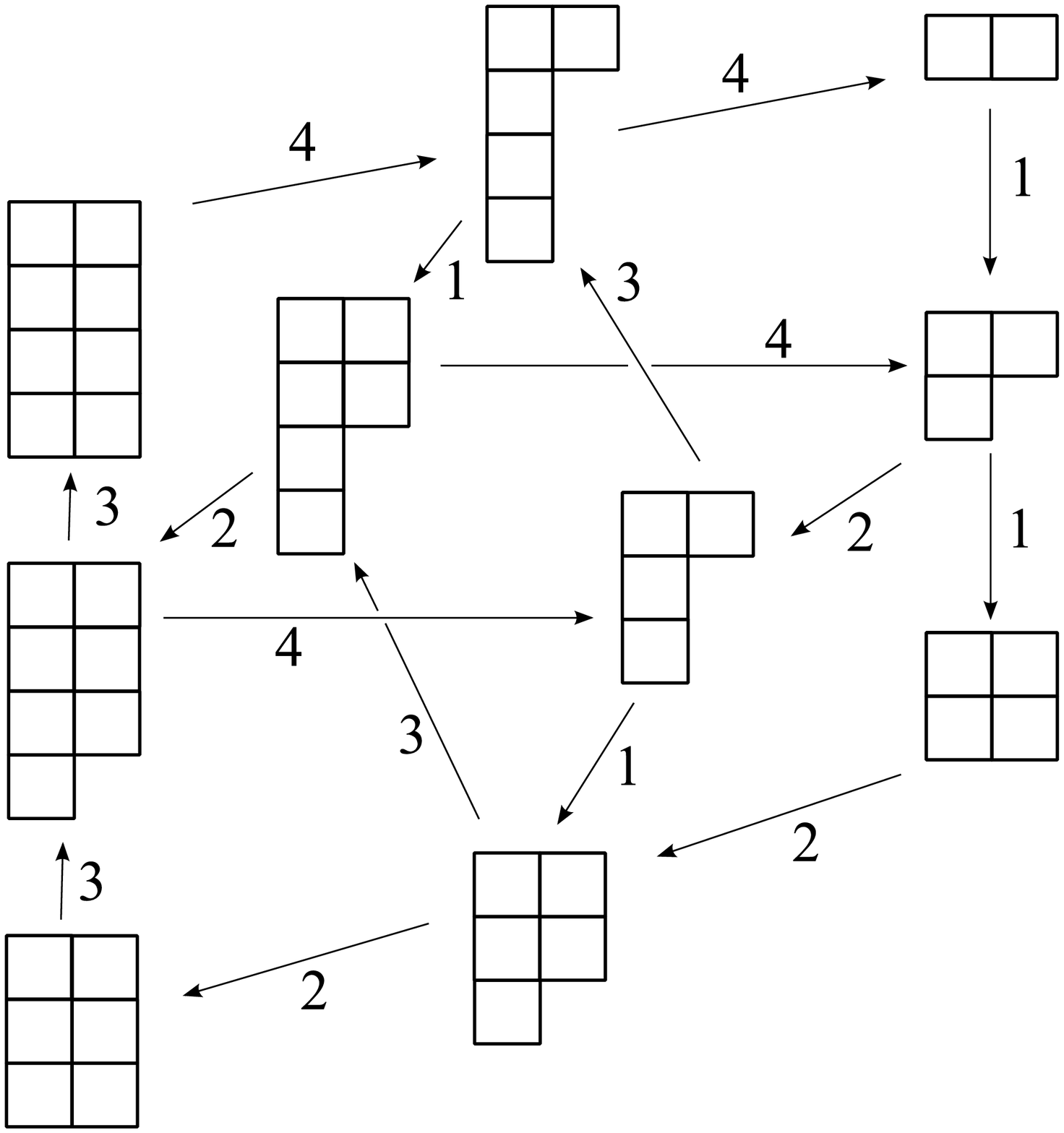}
\end{equation*}%
\caption{The crystal graph for $n=4,k=2$ and $z=1$. The vertices are the elements in $%
P_{k}^{+}$. Two vertices $\hat{\protect\lambda},\hat{\protect\mu}$ are
connected via an edge of colour $i$, if $\hat{\protect\mu}=a_{i}\hat{\protect%
\lambda}$.}
\label{fig:crystal}
\end{figure}

Consider the tensor product $V^{\otimes k}\cong V(zq^{-k+1})\otimes
V(zq^{-k+3})\otimes \cdots \otimes V(zq^{k-1})$ and let $V_{n,k}$ denote the
subspace invariant under the natural action of the Hecke algebra on $%
V^{\otimes k}$; see e.g. \cite[Chapter 10.2]{CP} and references therein.
There is a distinguished basis $\mathcal{B}_{n,k}=\{v_{\hat{\lambda}}\}_{%
\hat{\lambda}\in P_{k}^{+}}\subset V_{n,k}$ such that the pair $(\mathcal{L}%
_{n,k},\mathcal{B}_{n,k})$ with $\mathcal{L}_{n,k}=\tbigoplus_{\hat{\lambda}%
\in P_{k}^{+}}\mathbb{A}v_{\hat{\lambda}}$ forms a crystal basis of $V_{n,k}$
and the plactic generators $a_{i}^{\ast },a_{i}$ coincide with Kashiwara's
crystal operators $\tilde{E}_{i},\tilde{F}_{i}$, respectively. (We refer the
reader to \cite{HK} for an explanation of these terms.) For instance, let $%
\hat{\lambda}\in P_{k}^{+}$ be a strict partition, all parts are mutually
distinct, and denote by $\hat{\lambda}^{t}$ its transpose. Then the
corresponding basis vector in $V^{\otimes k}$ is given by%
\begin{equation*}
v_{\hat{\lambda}}=\frac{1}{[k]_{q}!}\sum_{\sigma \in S_{k}}q^{-\ell (\sigma
)}v_{\hat{\lambda}_{\sigma (k)}^{t}}\otimes \cdots \otimes v_{\hat{\lambda}%
_{\sigma (1)}^{t}}\;.
\end{equation*}%
In particular, $\mathcal{B}_{n,k}\cong P_{k}^{+}$ are identical as sets and
under the action of the local affine plactic algebra $P_{k}^{+}$ can be
viewed as a coloured, oriented graph (called the crystal graph of $V_{n,k}$%
). The vertices of this graph are the elements in $P_{k}^{+}$ and the edges
of colour $i$, $\hat{\lambda}\overset{i}{\longrightarrow }\hat{\mu}$, are
given by the relation $\hat{\mu}=a_{i}\hat{\lambda}$, where $a_{i}$ adds a
box in the $(i+1)^{\text{th}}$ row (if allowed) for $i=1,\ldots ,n-1$. The
letter $a_{n}$ removes a column of height $n$ and adds a box in the first
row if possible.

\begin{example}
\textrm{Let $n=4,\;k=2$ and for simplicity set $z=1$. Then{\small
\begin{equation*}
\Yvcentermath1P_{k}^{+}=\left\{ \mathbb{~}\yng(2)~,\;\yng(2,1)~,\;\yng%
(2,2)~,\;\yng(2,1,1)~,\;\yng(2,2,1)~,\;\yng(2,2,2)~,\;\yng(2,1,1,1)~,\;\yng%
(2,2,1,1)~,\;\yng(2,2,2,1),\;\yng(2,2,2,2)~\right\}
\end{equation*}%
}The corresponding basis vectors in $V^{\otimes k}$ for the first 5 elements
are{\normalsize \
\begin{equation*}
v_{1}\otimes v_{1},~\tfrac{v_{2}\otimes v_{1}-q^{-1}v_{1}\otimes v_{2}}{%
[2]_{q}},~v_{2}\otimes v_{2},~\tfrac{v_{3}\otimes v_{1}-q^{-1}v_{1}\otimes
v_{3}}{[2]_{q}},~\tfrac{v_{3}\otimes v_{2}-q^{-1}v_{2}\otimes v_{3}}{[2]_{q}}%
,\ldots
\end{equation*}%
} et cetera. The crystal graph resulting from the action of the local affine
plactic algebra on $P_{k}^{+}$ is depicted in Figure \ref{fig:crystal} for $z=1$. }
\end{example}

Note the difference in role played by the local affine plactic algebra and
the phase algebra. While the plactic algebra $\func{Pl}(\mathcal{A})$
preserves the level $k$ and describes the crystal structure of $\mathcal{H}%
_{k}$ with respect to the quantum affine algebra $U_{q}\widehat{\mathfrak{sl}%
}(n)$, the phase algebra $\hat{\Phi}$ increases and decreases the level $k$,
where the maps $\varphi _{i},\varphi _{i}^{\ast }:\mathcal{H}_{k}\rightarrow
\mathcal{H}_{k\mp 1}$ correspond to the crystal limit of the $U_{q}\widehat{%
\mathfrak{sl}}(2)$ generators. They generate a \textquotedblleft
tower\textquotedblright\ of crystals, the simplest example, $n=3$, is
depicted in Figure \ref{fig:tower}.

%For any finite dimensional module $\mathfrak{M}$ one can define its
%affinization $\mathfrak{M}^{\text{aff}}=\mathbb{Q}(q)[z,z^{-1}]\otimes _{%
%\mathbb{Q}(q)}\mathfrak{M}$ by setting%
%\begin{equation}
%E_{i}(z^{d}\otimes v)=z^{d+\delta _{in}}\otimes E_{i}v,\qquad
%F_{i}(z^{d}\otimes v)=z^{d-\delta _{in}}\otimes F_{i}v\qquad \text{and\qquad
%}K_{i}(z^{d}\otimes v)=z^{d}\otimes K_{i}v\;.
%\end{equation}%
%Let $\zeta \in \mathbb{Q}(q)$ be nonzero then the quotient of $\mathbb{Q}%
%(q)[z,z^{-1}]$ by the maximal ideal generated by $z-\zeta $ induces an
%isomorphism of fields $\mathbb{Q}(q)[z,z^{-1}]/\symbol{126}\rightarrow
%\mathbb{Q}(q)$ where $z\mapsto \zeta $. Specializing $\mathfrak{M}^{\text{aff%
%}}$ at $z=\zeta $ yields the corresponding evaluation module $\mathfrak{M}%
%(\zeta )$. If $\mathfrak{M}$ possesses the crystal basis $(\mathcal{L},%
%\mathcal{B})$ then $\mathfrak{M}^{\text{aff}}$ has the crystal basis $(%
%\mathcal{L}^{\text{aff}},\mathcal{B}^{\text{aff}})$ with
%\begin{equation}
%\mathcal{L}^{\text{aff}}=\tbigoplus_{k}\mathbb{A}[z,z^{-1}]\otimes _{\mathbb{%
%A}}\mathcal{L\qquad }\text{and}\mathcal{\qquad B}^{\text{aff}}=\{b[d]~|~b\in
%\mathcal{B},\;d\in \mathbb{Z}\}\;.
%\end{equation}%
%The action of Kashiwara's crystal operators is now given by%
%\begin{equation}
%\tilde{E}_{i}b[d]:=(\tilde{E}_{i}b)[d+\delta _{in}]\qquad \text{and}\qquad
%\tilde{F}_{i}b[d]:=(\tilde{F}_{i}b)[d-\delta _{in}]\;.
%\end{equation}
\begin{figure}[tbp]
\begin{equation*}
\includegraphics[scale=0.55]{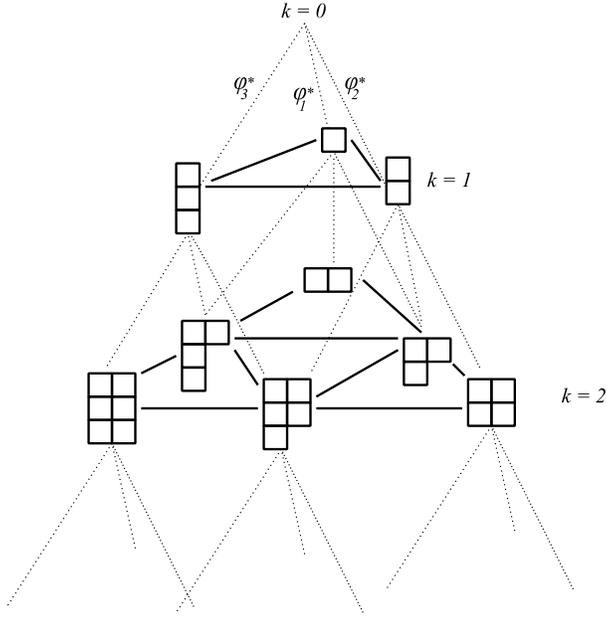}
\end{equation*}%
\caption{The tower of crystal graphs generated by the phase algebra for $n=3$
and $k=0,1,2$. The crystals for $k>1$ all consist of triangles which are
subdivided into smaller triangles similar as depicted for $k=2$. The dotted
lines correspond to the action of the phase algebra generators $\protect%
\varphi _{i}^{\ast }$ and the solid lines to the action of the plactic
algebra $\func{Pl}(\mathcal{A)}$.}
\label{fig:tower}
\end{figure}

\section{R-matrices in the crystal limit}

We are now going to define an integrable vertex model on $\mathcal{H}=%
\mathcal{M}^{\otimes n}$ along the same lines as the models discussed in
\cite{Baxter}. In the appendix it is shown that it arises from taking the
crystal limit of the $U_{q}\widehat{\mathfrak{sl}}(2)$ intertwiner,
generically called $R$-matrix, for the tensor product $M_{\mu }(u)\otimes
M_{\nu }(v)$ of the module discussed above.

\subsection{Definition of the vertex model}

\label{vertexdef}

Consider a $n\times n^{\prime }$ square lattice with periodic boundary
conditions in the horizontal direction. On the edges of the square lattice
live statistical variables $m\in \mathbb{Z}_{\mathbb{\geq }0}$, which we
identify with the basis vectors $\{v_{m}\}_{m\in \mathbb{Z}_{\mathbb{\geq }%
0}}$ in $\mathcal{M}$. Then a row configuration in the lattice, i.e. an
assignment of statistical variables $\boldsymbol{m}=(m_{1},\ldots ,m_{n})$
along one row of vertical edges, fixes a vector $v_{m_{1}}\otimes \cdots
\otimes v_{m_{n}}\in \mathcal{H}=\mathcal{M}^{\otimes n}$. Similarly, a
fixed lattice configuration of the entire lattice can be seen as a vector in
$\mathcal{H}^{\otimes n^{\prime }}$. Not each lattice configuration is
allowed, we single out particular ones by assigning to each local
configuration around a single vertex a \textquotedblleft Boltzmann
weight\textquotedblright\ (a pseudo-probability)
\begin{equation}
\mathcal{R}_{c,d}^{a,b}(u)=\left\{
\begin{array}{cc}
u^{a}, & d=a+b-c,\;b\geq c \\
0, & \text{else}%
\end{array}%
\right. \;,  \label{Boltzmannweight}
\end{equation}%
where $a,b,c,d\in \mathbb{Z}_{\mathbb{\geq }0}$ are the statistical
variables; see Figure \ref{fig:boltzmann}.

\begin{figure}[tbp]
\begin{equation*}
\includegraphics[scale=0.3]{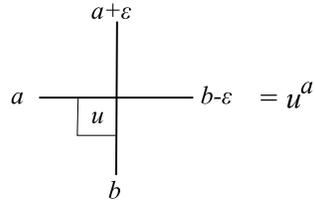}
\end{equation*}%
\caption{Graphical depiction of a vertex configuration with Boltzmann weight
(\protect\ref{Boltzmannweight}). The statistical variables are constrained
by $a,b,\protect\varepsilon=d-a,c=b-\protect\varepsilon\in\mathbb{Z}_{\geq
0} $.}
\label{fig:boltzmann}
\end{figure}

The variable $u$ is called the spectral parameter and we define $\mathcal{M}%
(u)=\mathbb{C}[u,u^{-1}]\otimes _{\mathbb{C}}\mathcal{M}$. The Boltzmann
weights define an operator $\mathcal{R}(u/v):\mathcal{M}(u)\otimes \mathcal{M%
}(v)\rightarrow \mathcal{M}(u)\otimes \mathcal{M}(v)$ via the relation
\begin{equation}
\mathcal{R}(u)~v_{a}\otimes v_{b}=\sum_{c,d\geq 0}\mathcal{R}%
_{c,d}^{a,b}(u)~v_{c}\otimes v_{d},  \label{Rmatrixdef}
\end{equation}%
which we can express in terms of the phase algebra generators (\ref%
{crystalNEF}): let $\mathcal{P}(v_{a}\otimes v_{b})=v_{b}\otimes v_{a}$ be
the flip operator then%
\begin{equation}
\mathcal{R}(u)=\mathcal{P}\left[ \tsum_{\alpha \in \mathbb{Z}_{\geq
0}}(\varphi ^{\ast })^{\alpha }\otimes \varphi ^{\alpha }\right]
(u^{N}\otimes 1)\;.  \label{crystalR}
\end{equation}%
Despite the infinite sum this operator is well-defined, since when acting on
an arbitrary vector $v_{a}\otimes v_{b}\in \mathcal{M}(u)\otimes \mathcal{M}%
(v)$ only a finite number of terms in the sum are nonzero.

Following the standard procedure \cite{Baxter} we now employ the $\mathcal{R}
$-matrix to define the discrete evolution operator, the row-to-row transfer
matrix, of our statistical mechanics model. Its matrix elements are obtained
by fixing two sets $\boldsymbol{m},\boldsymbol{m}^{\prime }\in \mathcal{H}=%
\mathcal{M}^{\otimes n}$ of statistical variables along the incoming and
outgoing vertical edges of a lattice row and summing over those variables
which sit at the horizontal edges; see Figure \ref{fig:rowconfig} for an
allowed row configuration. As an operator the transfer matrix is given by%
\begin{equation}
Q(u)=\limfunc{Tr}_{0}\left[ z^{N\otimes 1}\mathcal{R}_{0n}(u)\cdots \mathcal{%
R}_{01}(u)\right] \in \limfunc{End}\mathcal{H}[u,z],\;  \label{Qmatrix}
\end{equation}%
where $\mathcal{H}[u,z]:=\mathbb{C}(u,z)\otimes _{\mathbb{C}}\mathcal{H}$
and $z$ is an indeterminate (as we will see below the same as the one in
Proposition \ref{faithfulness}). The lower indices $i=1,\ldots ,n$ refer to
the $n$ vertical edges in one row (the $n$ factors of $\mathcal{M}$ in $%
\mathcal{H}$) and the index $0$ belongs to the horizontal edges over which
the sum is taken (the copy of $\mathcal{M}$ over which the trace is
computed). The trace together with the additional operator $z^{N\otimes 1}$
enforces quasi-periodic boundary conditions in the horizontal lattice
direction.

\begin{figure}[tbp]
\begin{equation*}
\includegraphics[scale=0.3]{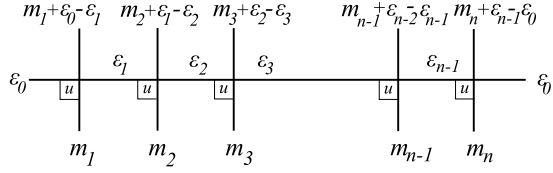}
\end{equation*}%
\caption{The allowed row configurations of the vertex model (\protect\ref%
{Boltzmannweight}) with $\protect\varepsilon_i,m_i,m_i-\protect\varepsilon%
_i\in\mathbb{Z}_{\geq 0}$. Due to the periodic boundary conditions $\protect%
\varepsilon_n =\protect\varepsilon_0$.}
\label{fig:rowconfig}
\end{figure}

While the trace is taken over an infinite-dimensional vector space, $%
\mathcal{M}$, for any pair of configurations $\boldsymbol{m},\boldsymbol{m}%
^{\prime }\in \mathcal{H}=\mathcal{M}^{\otimes n}$ only a finite number of
the terms making up the matrix element $\langle \boldsymbol{m},Q(u)%
\boldsymbol{m}^{\prime }\rangle $ is non-zero. The operator $Q$ is therefore
well-defined. We will show this in the proof of Proposition \ref{main} below
by explicitly computing the matrix elements and showing that $Q(u)$ should
be understood as formal power series in $u$ with operator valued
coefficients. First we wish to show integrability of our vertex model, i.e.
that the operator $Q$ commutes with itself for any pair of spectral
parameters.

\begin{proposition}
Let $\mathcal{R}$ be the operator (\ref{crystalR}) and define $\mathcal{S}%
(u/v)\in \limfunc{End}[\mathcal{M}(u)\otimes \mathcal{M}(v)]$ by setting%
\begin{equation}
\mathcal{S}(u)=(1-u)\mathcal{R}(u)+\mathcal{P}(u^{N+1}\otimes 1)\;.
\label{crystalS}
\end{equation}%
Then we have the identity%
\begin{equation}
\mathcal{S}_{12}(u)\mathcal{R}_{13}(uv)\mathcal{R}_{23}(v)=\mathcal{R}%
_{23}(v)\mathcal{R}_{13}(uv)\mathcal{S}_{12}(u),  \label{YBE}
\end{equation}%
where the lower indices indicate in which factor of the tensor product $%
\mathcal{M}(u)\otimes \mathcal{M}(uv)\otimes \mathcal{M}(v)$ the respective
operators act non-trivially. Moreover, $\mathcal{S}$ is invertible, $%
\mathcal{S}^{-1}(u)=\mathcal{PS}(u^{-1})\mathcal{P}$.
\end{proposition}

\begin{proof}
Because of the explicit appearance of the flip operator $\mathcal{P}$ in $%
\mathcal{R},\mathcal{S}$ it is convenient to work with $\mathcal{\hat{R}}=%
\mathcal{PR}$ and $\mathcal{\hat{S}}=\mathcal{P\hat{S}}$. Then the
Yang-Baxter equation is rewritten as follows%
\begin{equation*}
\lbrack 1\otimes \mathcal{\hat{S}}(u)][\mathcal{\hat{R}}(uv)\otimes
1][1\otimes \mathcal{\hat{R}}(v)]=[\mathcal{\hat{R}}(v)\otimes 1][1\otimes
\mathcal{\hat{R}}(uv)][\mathcal{\hat{S}}(u)\otimes 1]\;.
\end{equation*}%
We now compute the corresponding identity in terms of matrix elements by
evaluating the identity on the vector $|a,b,c\rangle \equiv v_{a}\otimes
v_{b}\otimes v_{c}$ and then multiplying from the left with the dual vector $%
\langle d,e,f|=v_{d}^{\ast }\otimes v_{e}^{\ast }\otimes v_{f}^{\ast }$.
Both sides of the identity vanish (and it therefore holds trivially true)
unless $c-f,$~$d-a\geq 0$ and $a+b+c=d+e+f$ according to (\ref%
{Boltzmannweight}). Provided these conditions are satisfied, the computation
of the left hand side yields,%
\begin{multline*}
\langle d,e,f|[1\otimes \mathcal{\hat{S}}(u)][\mathcal{\hat{R}}(uv)\otimes
1][1\otimes \mathcal{\hat{R}}(v)]|a,b,c\rangle = \\
(uv)^{a+b}(1-u)\sum_{i=\max (0,d-a-b)}^{c-f}u^{i-d+a}+(uv)^{a+b}u^{1+e-b},
\end{multline*}%
while from the right hand side we obtain%
\begin{multline*}
\langle d,e,f|[\mathcal{\hat{R}}(v)\otimes 1][1\otimes \mathcal{\hat{R}}%
(uv)][\mathcal{\hat{S}}(u)\otimes 1]|a,b,c\rangle = \\
(uv)^{a+b}(1-u)\sum_{i=0}^{\min (b,d-a)}u^{-i}+(uv)^{a+b}u\;.
\end{multline*}%
Distinguishing the two cases $b\geq d-a$ and $b<d-a$ we now verify that the
identity is true.

The formula for the inverse is verified in a similar manner, by acting with $%
\mathcal{S}$ on the subspace spanned by the vectors $v_{a}\otimes v_{b}$
with $a+b$ fixed. This subspace is finite-dimensional and the above formula
can be explicitly computed.
\end{proof}

\begin{corollary}[Integrability]
\label{integrability} The transfer matrix $Q$ of our vertex model commutes
for any pair $(u,v)$ of spectral parameters, $[Q(u),Q(v)]=0$, and, hence,
the model is integrable.
\end{corollary}

\begin{proof}
After writing out the product $Q(u)Q(v)$ as a trace over $\mathcal{R}$-matrices
according to the definition (\ref{Qmatrix}) and inserting the identity $1=%
\mathcal{S}(u)\mathcal{S}^{-1}(u)$ under the trace the assertion follows
from the Yang-Baxter equation (\ref{YBE}).
\end{proof}

\subsection{Relation with the phase model}

To conclude this section we explain how the vertex model (\ref%
{Boltzmannweight}) is related to the phase model of Bogoliubov, Izergin and Kitanine \cite%
{Bogoliubovetal}. The phase model was used in \cite[Part I]{KS} to give a
combinatorial description of the $\widehat{\mathfrak{sl}}(n)_{k}$ WZNW
fusion ring. As we will show in the next sections we arrive at the same
description of the fusion ring via the transfer matrix (\ref{Qmatrix}) which
we obtained from the $U_{q}\widehat{\mathfrak{sl}}(2)$-vertex model in the
crystal limit. We are therefore lead to investigate the relation between our
vertex model and the phase model. We recall that the $L$-operator of the
phase model is given by (cf. \cite{Bogoliubovetal}, \cite[Sections 3.2 and 4]%
{KS})%
\begin{equation}
L(u)=\sigma ^{+}\sigma ^{-}\otimes 1+\sigma ^{+}\otimes \varphi +u~\sigma
^{-}\otimes \varphi ^{\ast }+u~\sigma ^{-}\sigma ^{+}\otimes 1\;,
\label{phaseL}
\end{equation}%
where %
%\begin{equation}
$\sigma ^{+}=\left(
\begin{smallmatrix}
0 & 1 \\
0 & 0%
\end{smallmatrix}%
\right) $ and % \qquad \text{and}\qquad
$\sigma ^{-}=\left(
\begin{smallmatrix}
0 & 0 \\
1 & 0%
\end{smallmatrix}%
\right) $ %\label{Pauli}
%\end{equation}%
are the Pauli matrices acting in $\mathbb{C}(u)^{2}$. The transfer matrix
for the phase model then reads in analogy with (\ref{Qmatrix}),%
\begin{equation}
T(u)=\limfunc{Tr}_{0}z^{\sigma ^{3}\otimes 1}L_{0n}(u)\cdots L_{01}(u)\in
\limfunc{End}\mathcal{H}[u,z],  \label{Tmatrix}
\end{equation}%
where $\sigma ^{3}=\sigma ^{-}\sigma ^{+}$ and the
so-called auxiliary space indexed by \textquotedblleft 0\textquotedblright ,
over which the trace is taken, is now $\mathbb{C}^{2}(u)$. The following
proposition is obtained from a straightforward computation.

\begin{proposition}
Define the following element in $\limfunc{End}(\mathbb{C}(u)^{2}\otimes
\mathcal{M})$,%
\begin{equation}
L^{\prime }(u)=L(u)+u~\sigma ^{+}\sigma ^{-}\otimes (1-\varphi ^{\ast
}\varphi ),  \label{phaseL'}
\end{equation}%
then we have in $\limfunc{End}[\mathbb{C}(u)^{2}\otimes \mathcal{M}%
(v)\otimes \mathcal{M}]$ the identity%
\begin{equation}
L_{12}^{\prime }(u/v)L_{13}(u)\mathcal{R}_{23}(v)=\mathcal{R}%
_{23}(v)L_{13}(u)L_{12}^{\prime }(u/v)~.  \label{phaseYBE}
\end{equation}
\end{proposition}

\begin{remark}
\textrm{Note that the operator $L^{\prime }$ does not possess an inverse.
Similarly, like the other operators (\ref{crystalR}) and (\ref{crystalS}) it
is another special crystal limit of the $U_{q}\widehat{\mathfrak{sl}}(2)$%
-intertwiner associated with $M_{\mu }(u)\otimes M_{\nu }(v)$; see the
appendix.}
\end{remark}

Within the context of the quantum inverse scattering method \cite{KBI} one
introduces the Yang-Baxter algebra which for the phase model is generated by
the following matrix elements%
\begin{equation}
\left(
\begin{array}{cc}
A(u) & B(u) \\
C(u) & D(u)%
\end{array}%
\right) :=\langle v_{\sigma ^{\prime }}^{\ast }\otimes 1,z^{\sigma
^{3}\otimes 1}L_{0n}(u)\cdots L_{01}(u)v_{\sigma }\otimes 1\rangle _{\sigma
^{\prime },\sigma =0,1}\;.  \label{phaseYBalgebra}
\end{equation}%
The matrix elements $A,B,C,D\in \limfunc{End}\mathcal{H}[u,z]$ have a
particularly simple combinatorial action; cf. \cite[Cor 3.9]{KS}. Note that $%
T(u)=A(u)+zD(u)$. In analogy we now define for our vertex model (\ref%
{Qmatrix}) the infinite-dimensional operator-valued matrix%
\begin{equation}
Q_{\varepsilon ^{\prime },\varepsilon }(u):=\langle v_{\varepsilon ^{\prime
}}^{\ast }\otimes 1,z^{N\otimes 1}\mathcal{R}_{0n}(u)\cdots \mathcal{R}%
_{01}(u)v_{\varepsilon }\otimes 1\rangle \in \limfunc{End}\mathcal{H}[u,z]
\label{infYBalgebra}
\end{equation}%
with $\varepsilon ,\varepsilon ^{\prime }\in \mathbb{Z}_{\geq 0}$ and $%
Q(u)=\sum_{\varepsilon \geq 0}Q_{\varepsilon ,\varepsilon }(u)$. We compute
these matrix elements explicitly in the next section and show that also they
have a nice combinatorial interpretation. First we have the following
consequence from the previous proposition.

\begin{corollary}
The generators of the Yang-Baxter algebra for the phase model (\ref{Tmatrix}%
) and the vertex model (\ref{Qmatrix}) obey the following commutation
relations:%
\begin{gather}
A(u)Q_{\varepsilon ^{\prime },\varepsilon }(v)-Q_{\varepsilon ^{\prime
},\varepsilon }(v)A(u) = Q_{\varepsilon ^{\prime },\varepsilon
-1}(v)B(u)-u/v~C(u)Q_{\varepsilon ^{\prime }-1,\varepsilon }(v)
\label{Acomm} \\
\hspace{5cm}+\frac{u}{v}\left[ \delta _{\varepsilon ,0}Q_{\varepsilon
^{\prime },0}(v)A(u)-\delta _{\varepsilon ^{\prime },0}A(u)Q_{0,\varepsilon
}(v)\right]  \notag \\
Q_{\varepsilon ^{\prime },\varepsilon }(v)B(u)- v/u~B(u)Q_{\varepsilon
^{\prime },\varepsilon }(v) = D(u)Q_{\varepsilon ^{\prime }-1,\varepsilon
}(v)-Q_{\varepsilon ^{\prime },\varepsilon +1}(v)A(u)  \label{Bcomm} \\
\hspace{7cm}+\delta _{\varepsilon ^{\prime},0}B(u)Q_{0,\varepsilon }(v)
\notag \\
Q_{\varepsilon ^{\prime },\varepsilon }(v)C(u)- u/v~C(u)Q_{\varepsilon
^{\prime },\varepsilon }(v) = A(u)Q_{\varepsilon ^{\prime }+1,\varepsilon
}(v) -Q_{\varepsilon ^{\prime },\varepsilon -1}(v)D(u)  \label{Ccomm} \\
\hspace{7cm}-\delta _{\varepsilon ,0}\frac{u}{v}Q_{\varepsilon ^{\prime
},0}(v)C(u)  \notag \\
D(u)Q_{\varepsilon ^{\prime },\varepsilon }(v)- Q_{\varepsilon ^{\prime
},\varepsilon }(v)D(u) =Q_{\varepsilon ^{\prime },\varepsilon
+1}(v)C(u)-v/u~B(u)Q_{\varepsilon ^{\prime }+1,\varepsilon }(v)
\label{Dcomm}
\end{gather}%
Here matrix elements with negative indices are understood to be zero. In
particular, we have%
\begin{equation}
D(u)Q_{\varepsilon ^{\prime },\varepsilon }(v)-Q_{\varepsilon ^{\prime
},\varepsilon }(v)D(u)=Q_{\varepsilon ^{\prime }+1,\varepsilon
+1}(v)A(u)-A(u)Q_{\varepsilon ^{\prime }+1,\varepsilon +1}(v)\;.
\label{phasevertexcomm}
\end{equation}
\end{corollary}

\begin{proof}
Set $\mathcal{T}(u)=z^{\sigma ^{3}\otimes 1}L_{0n}(u)\cdots L_{01}(u)$ and $%
\mathcal{Q}(v)=z^{N\otimes 1}\mathcal{R}_{0n}(v)\cdots \mathcal{R}_{01}(v)$.
Then the first four commutation relations are easily obtained by considering
the following equality of matrix elements%
\begin{multline*}
\langle v_{\sigma ^{\prime }}^{\ast }\otimes v_{\varepsilon ^{\prime
}}^{\ast }\otimes 1,L_{12}^{\prime }(u/v)\mathcal{T}_{13}(u)\mathcal{Q}%
_{23}(v)v_{\sigma }^{\ast }\otimes v_{\varepsilon }\otimes 1\rangle = \\
\langle v_{\sigma ^{\prime }}^{\ast }\otimes v_{\varepsilon ^{\prime
}}^{\ast }\otimes 1,\mathcal{Q}_{23}(v)\mathcal{T}_{13}(u)L_{12}^{\prime
}(u/v)v_{\sigma }^{\ast }\otimes v_{\varepsilon }\otimes 1\rangle
\end{multline*}%
which follows from (\ref{phaseYBE}). For the last identity (\ref%
{phasevertexcomm}) employ the equality (\ref{Dcomm}) together with the
relations%
\begin{equation}
C(u)=\varphi _{n}A(u)\qquad \text{and}\qquad B(u)=uA(u)\varphi _{1}^{\ast },
\label{CABA}
\end{equation}%
which follow from the definition (\ref{phaseL}) and (\ref{phaseYBalgebra}).
\end{proof}

\section{Noncommutative symmetric polynomials}

We now connect the definition of our infinite-dimensional vertex model in
terms of the transfer matrix (\ref{Qmatrix}) with the discussion of the
plactic algebra in Section \ref{plactic}. We show that the matrix elements
of $Q$ can be written as analogues of the complete symmetric functions in a
noncommutative alphabet, the local affine plactic algebra $\limfunc{Pl}(%
\mathcal{A})$ in the representation (\ref{placticrep}). Recall that given a
set of \emph{commutative} indeterminates $x=(x_{1},\ldots ,x_{n}),$ the
complete symmetric functions are the coefficients in the formal power series
expansion of the following generating function \cite[Chapter I, Section 2,
p21]{MacDonald},%
\begin{equation}
H(u)=\prod_{i=1}^{n}\frac{1}{1-x_{i}u}=\sum_{r\geq 0}h_{r}(x_{1},\ldots
,x_{n})u^{r}~.  \label{Hdef}
\end{equation}%
This definition implies the following recursive formula with respect to $n$,%
\begin{equation}
h_{r}(x_{1},\ldots ,x_{n})=h_{r}(x_{1},\ldots
,x_{n-1})+x_{n}h_{r-1}(x_{1},\ldots ,x_{n})\;.  \label{hrecursion}
\end{equation}%
The solution is given by the following explicit expression%
\begin{equation}
h_{r}(x_{1},\ldots ,x_{n})=\sum_{p\vdash r}x_{1}^{p_{1}}\cdots x_{n}^{p_{n}},
\label{hdef}
\end{equation}%
where the sum runs over all compositions $p$ of $r>0$ and $h_{0}=1$. Up to a
specific ordering we now show that the same formulae hold for a series
expansion of (\ref{Qmatrix}) when replacing the commutative variables $%
(x_{1},\ldots ,x_{n})$ with the generators of $\limfunc{Pl}(\mathcal{A})$,
i.e. the $Q$-operator is the generating function complete symmetric
polynomials in a noncommutative alphabet.

\begin{proposition}
\label{main} Let $Q(u)$ be the operator defined in (\ref{Qmatrix}) and
denote by $\varphi _{i},\varphi _{i}^{\ast }$ and $a_{i}=\varphi _{i}\varphi
_{i+1}^{\ast }$ the generators of the phase and local affine plactic
algebra; see Sections \ref{crystalverma} and \ref{plactic}. Then we have the
following formal power series expansion%
\begin{equation}
Q(u)=\limfunc{Tr}_{0}z^{N\otimes 1}\mathcal{R}_{0n}(u)\cdots \mathcal{R}%
_{01}(u)=\sum_{r\geq 0}u^{r}h_{r}(\mathcal{A}),  \label{ncHdef}
\end{equation}%
where
\begin{equation}
h_{r}(\mathcal{A}):=\sum_{|\boldsymbol{\varepsilon }|=r}z^{\varepsilon
_{0}}(\varphi _{1}^{\ast })^{\varepsilon _{0}}a_{1}^{\varepsilon
_{1}}a_{2}^{\varepsilon _{2}}\cdots a_{n-1}^{\varepsilon _{n-1}}\varphi
_{n}^{\varepsilon _{0}}  \label{nchdef}
\end{equation}%
with the sum running over all compositions $\boldsymbol{\varepsilon }%
=(\varepsilon _{0},\varepsilon _{1},\ldots ,\varepsilon _{n-1})$ and $|%
\boldsymbol{\varepsilon }|=\sum_{i}\varepsilon _{i}$. In particular, when
setting $z=0$ we obtain that $Q_{0,0}(u)=\sum_{r\geq 0}u^{r}h_{r}(\mathcal{A}%
^{\prime })$, where $h_{r}(\mathcal{A}^{\prime })=\sum_{|\boldsymbol{%
\varepsilon }|=r,\varepsilon _{0}=0}a_{1}^{\varepsilon
_{1}}a_{2}^{\varepsilon _{2}}\cdots a_{n-1}^{\varepsilon _{n-1}}$ are the
complete symmetric polynomials in the \emph{finite} plactic algebra $%
\limfunc{Pl}_{\func{fin}}(\mathcal{A}^{\prime })$ and analogous to (\ref%
{hrecursion}) we have the recursion relation%
\begin{equation}
h_{r}(\mathcal{A})=h_{r}(\mathcal{A}^{\prime })+z\varphi _{1}^{\ast }h_{r-1}(%
\mathcal{A})\varphi _{n}\;.  \label{nchrecursion}
\end{equation}
\end{proposition}

\begin{remark}
\textrm{To facilitate the comparison with the commutative case, assume that
there exists one summand in the sum (\ref{nchdef}) for which $\varepsilon
_{j}$ vanishes. Then the corresponding monomial can be rewritten as ($%
\varepsilon _{j}=0$)%
\begin{multline*}
z^{\varepsilon _{0}}(\varphi _{1}^{\ast })^{\varepsilon
_{0}}a_{1}^{\varepsilon _{1}}a_{2}^{\varepsilon _{2}}\cdots
a_{n-1}^{\varepsilon _{n-1}}\varphi _{n}^{\varepsilon _{0}}= \\
z^{\varepsilon _{0}}a_{j+1}^{\varepsilon _{j+1}}a_{j+2}^{\varepsilon
_{j+2}}\cdots a_{n-1}^{\varepsilon _{n-1}}\varphi _{n}^{\varepsilon
_{0}}(\varphi _{1}^{\ast })^{\varepsilon _{0}}a_{1}^{\varepsilon
_{1}}a_{2}^{\varepsilon _{2}}\cdots a_{j-1}^{\varepsilon _{j-1}}= \\
a_{j+1}^{\varepsilon _{j+1}}a_{j+2}^{\varepsilon _{j+2}}\cdots
a_{n-1}^{\varepsilon _{n-1}}a_{0}^{\varepsilon _{0}}a_{1}^{\varepsilon
_{1}}a_{2}^{\varepsilon _{2}}\cdots a_{j-1}^{\varepsilon _{j-1}},
\end{multline*}%
where we have exploited that $[a_{i},a_{j}]=0$ for $|i-j|\func{mod}n>1$ and $%
a_{n}=z\varphi _{n}\varphi _{1}^{\ast }$. In particular, if $r<n$ then we
can always find for each summand in (\ref{nchdef}) some $1\leq j\leq n$ such
that $\varepsilon _{j}=0$ and the definition of the complete symmetric
polynomials coincides with the one in \cite[Def 5.16]{KS},%
\begin{equation*}
r<n:\;h_{r}(\mathcal{A})=\sum_{p\vdash r}\prod_{i=1}^{\circlearrowright
}a_{i}^{p_{i}},\qquad \prod_{i=1}^{\circlearrowright
}a_{i}^{p_{i}}=a_{j+1}^{\varepsilon _{j+1}}\cdots a_{n-1}^{\varepsilon
_{n-1}}a_{0}^{\varepsilon _{0}}a_{1}^{\varepsilon _{1}}\cdots
a_{j-1}^{\varepsilon _{j-1}},
\end{equation*}%
where the letters are clockwise cyclically ordered. The similarity with (\ref%
{hdef}) is now apparent. The new result here, compared to \cite[Cor 6.9]{KS}%
, is the explicit expression for $h_{r}(\mathcal{A})$ when $r\geq n$. }
\end{remark}

\begin{proof}
The proof is immediate from the definition of the Boltzmann weights and the
action of the phase algebra. Namely, consider an allowed row configuration
as depicted in Figure \ref{fig:rowconfig}. Summing over the variables
located at the horizontal edges $\boldsymbol{\varepsilon }=(\varepsilon
_{0},\varepsilon _{1},\ldots ,\varepsilon _{n-1})$ we obtain the matrix
element
\begin{multline*}
\langle \boldsymbol{m}^{\prime }|Q(u)|\boldsymbol{m}\rangle
=\sum_{\varepsilon _{0}}\langle \boldsymbol{m}^{\prime }|Q_{\varepsilon
_{0},\varepsilon _{0}}(u)|\boldsymbol{m}\rangle \\
=\sum_{\boldsymbol{\varepsilon }}z^{\varepsilon _{0}}u^{\varepsilon
_{0}+\cdots +\varepsilon _{n-1}}\langle \boldsymbol{m}^{\prime }|(\varphi
_{1}^{\ast })^{\varepsilon _{0}}\varphi _{1}^{\varepsilon _{1}}(\varphi
_{2}^{\ast })^{\varepsilon _{1}}\varphi _{2}^{\varepsilon _{2}}\cdots
(\varphi _{n}^{\ast })^{\varepsilon _{n-1}}\varphi _{n}^{\varepsilon _{0}}|%
\boldsymbol{m}\rangle ,
\end{multline*}%
where $|\boldsymbol{m}|=|\boldsymbol{m}^{\prime }|=k\geq 0.$ If $|%
\boldsymbol{m}|\neq |\boldsymbol{m}^{\prime }|$ the matrix element vanishes
according to (\ref{Boltzmannweight}). Note in particular that because of the
ordering of the phase algebra generators it follows for $r>k$ that $h_{r}(%
\mathcal{A})\mathcal{H}_{k}=\{0\}$ and, thus, the series expansion
terminates after finitely many summands on each $\mathcal{H}_{k}$. The
operator (\ref{Qmatrix}) is therefore well defined as claimed earlier. The
result (\ref{ncHdef}) with (\ref{nchdef}) now follows from (\ref{placticrep}%
).
\end{proof}

\begin{remark}
\textrm{As shown in \cite[Prop 5.13]{KS} the transfer matrix (\ref{Tmatrix})
of the phase model is the generating function for the noncommutative
analogues of the elementary symmetric polynomials,%
\begin{equation}
T(u)=\sum_{r=0}^{n}u^{r}e_{r}(\mathcal{A}),\qquad e_{r}(\mathcal{A})=\sum
_{\substack{ p\vdash r  \\ p_{i}=0,1}}\prod_{i=1}^{\circlearrowleft
}a_{i}^{p_{i}}  \label{ncedef}
\end{equation}%
where the letters are now anticlockwise cyclically ordered and $e_{n}(%
\mathcal{A})=z\cdot 1$. Also in this instance the familiar recursion
relation from the commutative case, $e_{r}(x_{1},\ldots
,x_{n})=e_{r}(x_{1},\ldots ,x_{n-1})+x_{n}e_{r-1}(x_{1},\ldots ,x_{n-1})$
with generating function $E(u)=\prod_{i=1}^{n}(1+ux_{i})=\sum_{r\geq
0}u^{r}e_{r}(x_{1},\ldots ,x_{n})$,$\;$generalises to the noncommutative
case,%
\begin{equation}
e_{r}(\mathcal{A})=e_{r}(\mathcal{A}^{\prime })+z\varphi _{n}e_{r-1}(%
\mathcal{A}^{\prime })\varphi _{1}^{\ast }\;.  \label{ncerecursion}
\end{equation}%
This last equality is implicit in the results of \cite[Prop 5.13]{KS}.
Setting $z=0$ one infers that
\begin{equation*}
A(u)=(1+ua_{n-1})(1+ua_{n-2})\cdots (1+ua_{1})=\sum_{r\geq 0}u^{r}e_{r}(%
\mathcal{A}^{\prime })
\end{equation*}%
and the identity (\ref{ncerecursion}) then follows from (\ref{CABA}).}
\end{remark}

\begin{corollary}
The elementary and complete symmetric polynomials (\ref{nchdef}) in the
noncommutative alphabet $\mathcal{A}$ pairwise commute,
\begin{equation}
\lbrack e_{r}(\mathcal{A}),e_{r^{\prime }}(\mathcal{A})]\overset{(1)}{=}%
[h_{r}(\mathcal{A}),h_{r^{\prime }}(\mathcal{A})]\overset{(2)}{=}[e_{r}(%
\mathcal{A}),h_{r^{\prime }}(\mathcal{A})]\overset{(3)}{=}0\;.
\label{ncehcomm}
\end{equation}
\end{corollary}

\begin{proof}
The first equality is a result of \cite[Cor 5.14]{KS}. The second equality
in (\ref{ncehcomm}) is a direct consequence of Corollary \ref{integrability}%
. Finally, to prove the third equality we employ (\ref{phasevertexcomm}) to
arrive at%
\begin{equation}
T(u)Q(v)-Q(v)T(u)=Q_{0,0}(v)A(u)-A(u)Q_{0,0}(v)=0\;.
\end{equation}%
Thus, $[e_{r}(\mathcal{A}),h_{r^{\prime }}(\mathcal{A})]=-[e_{r}(\mathcal{A}%
^{\prime }),h_{r^{\prime }}(\mathcal{A}^{\prime })]=0$. That the last
commutator in the finite plactic algebra vanishes is a direct consequence of
(\ref{Acomm}) for $\varepsilon =\varepsilon ^{\prime }=0$.
\end{proof}

\begin{figure}[tbp]
\begin{equation*}
\includegraphics[scale=0.30]{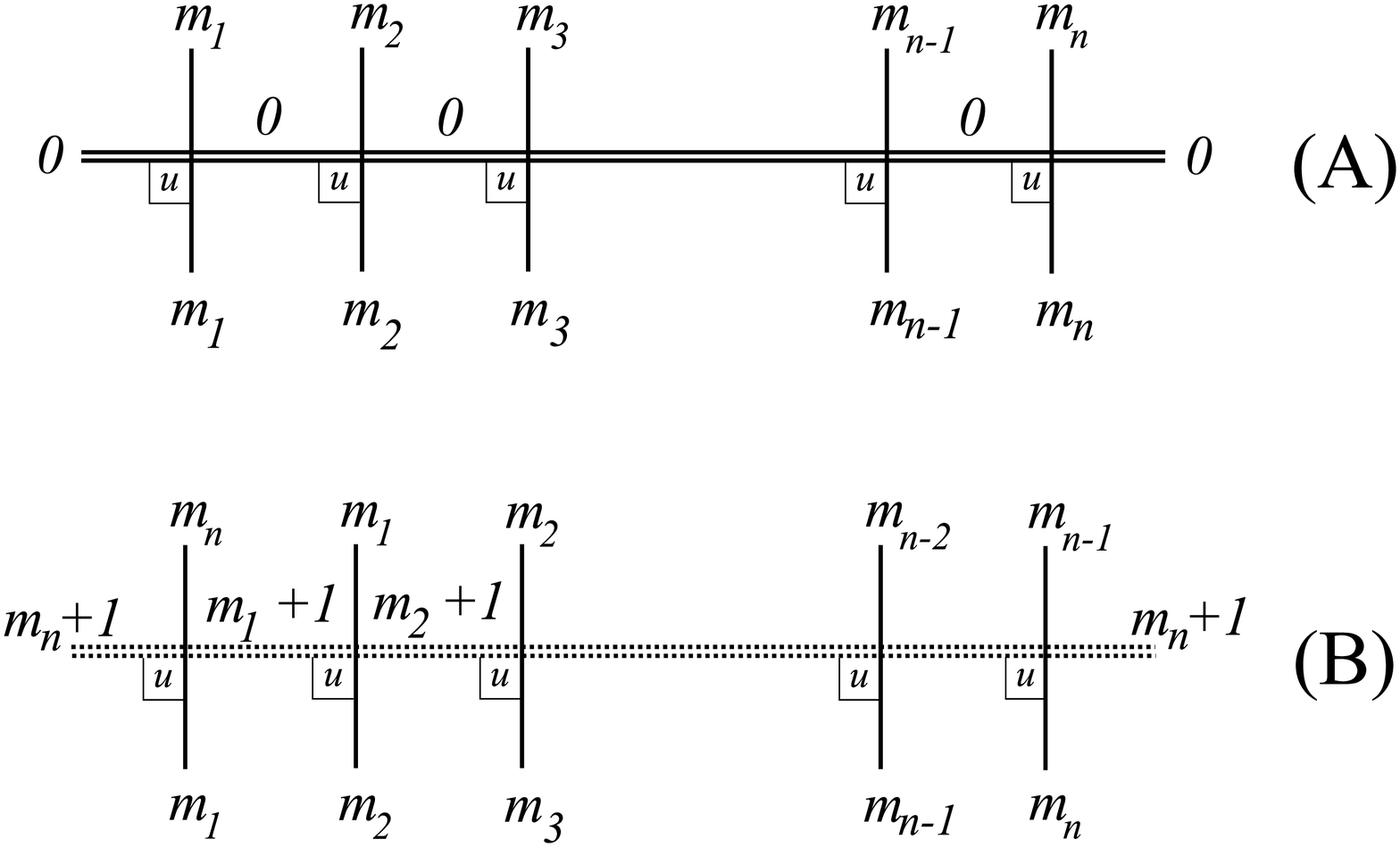}
\end{equation*}%
\caption{The allowed row configurations in the decomposition of the product $%
T(-u)Q(u)$; see the proof of Proposition \protect\ref{TQeqn}. Here the
double solid lines stand for the subspace $W$ and the double dashed lines
for the complement $\bar{W}$.}
\label{fig:rowconfig_TQ}
\end{figure}

\begin{proposition}[$TQ$-equation]
\label{TQeqn}Let $T$ and $Q$ be the transfer matrices (\ref{Tmatrix}) and (%
\ref{Qmatrix}), respectively. Then they satisfy the following identity%
\begin{eqnarray}
T(-u)Q(u) &=&\left[ Q(uq)+z(-u)^{n}q^{K}Q(uq^{-1})\right] _{q=0}
\label{crystalTQ} \\
&=&1+z(-u)^{n}\sum_{k\geq 0}u^{k}h_{k}(\mathcal{A})\pi _{k}\;,  \notag
\end{eqnarray}%
where $K=\sum_{i}N_{i}$ and $\pi _{k}$ is the (orthogonal) projector onto $%
\mathcal{H}_{k}\subset \mathcal{H}$.
\end{proposition}

\begin{proof}
Let $L^{\prime }$ be the operator defined in (\ref{phaseL'}). Then $W=\ker
L^{\prime }(-1)\subset \mathbb{C}^{2}\otimes \mathcal{M}$ and the complement
$\bar{W}\cong (\mathbb{C}^{2}\otimes \mathcal{M})/W$ are spanned by the
vectors%
\begin{equation*}
w_{m}:=\left\{
\begin{array}{cc}
v_{0}\otimes v_{0}, & m=0 \\
v_{0}\otimes v_{m}+v_{1}\otimes v_{m-1}, & m>0%
\end{array}%
\right. \qquad \text{and}\qquad \bar{w}_{m}:=v_{0}\otimes v_{m+1}\;,
\end{equation*}%
respectively. From (\ref{phaseYBE}) we infer that $L_{13}(-u)\mathcal{R}%
_{23}(u)W\otimes \mathcal{M}\subset W\otimes \mathcal{M}$. In fact, we have
that%
\begin{equation}
L_{13}(-u)\mathcal{R}_{23}(u)w_{m}\otimes v_{a}=\delta
_{m,0}\sum_{b=0}^{a}w_{b}\otimes v_{a-b}  \label{LRW}
\end{equation}%
and
\begin{equation}
L_{13}(-u)\mathcal{R}_{23}(u)\bar{w}_{m}\otimes v_{a}=-u^{m}\bar{w}%
_{a}\otimes v_{m-1}+~\ldots ~,  \label{LRWbar}
\end{equation}%
where the omitted terms in the second equality lie in $W$. Thus, we may write%
\begin{multline*}
T(-u)Q(u)=\limfunc{Tr}_{\mathbb{C}^{2}\otimes \mathcal{M}}[z^{\sigma^3\otimes
N\otimes 1}L_{0n}(-u)\mathcal{R}_{0^{\prime }n}(u)\cdots L_{01}(-u)\mathcal{R%
}_{0^{\prime }1}(u)]= \\
\limfunc{Tr}_{W}[z^{\sigma^3\otimes N\otimes 1}L_{0n}(-u)\mathcal{R}_{0^{\prime
}n}(u)\cdots L_{01}(-u)\mathcal{R}_{0^{\prime }1}(u)] \\
+\limfunc{Tr}_{\bar{W}}[z^{\sigma^3\otimes N\otimes 1}L_{0n}(-u)\mathcal{R}%
_{0^{\prime }n}(u)\cdots L_{01}(-u)\mathcal{R}_{0^{\prime }1}(u)],
\end{multline*}%
where the indices $0$ and $0^{\prime }$ refer to the factor $\mathbb{C}^{2}$
and $\mathcal{M}$ in $\mathbb{C}^{2}\otimes \mathcal{M}$, respectively. The
assertion now follows by observing that (\ref{LRW}) and (\ref{LRWbar}) imply
that the only allowed vertex configurations in a row are the ones depicted
in Figure \ref{fig:rowconfig_TQ}. The configuration labelled (A) corresponds
to the trace over $W$ and yields the identity in (\ref{crystalTQ}).

The second configuration (B) describes the trace over $\bar{W}$ and
coincides with the action of the $\widehat{\mathfrak{sl}}(n)$-Dynkin diagram
automorphism, $\limfunc{rot}:\boldsymbol{m}\mapsto (m_{n},m_{1},\ldots
,m_{n-1})$, which for $z=1$ is identical with the action of $h_{k}(\mathcal{A%
})$; see (\ref{nchdef}).
\end{proof}

The following result is contained in \cite[Part I, Def 5.16 and Cor 6.9]{KS}
for $r<n$. Here we state an alternative proof valid for all $r>0$.

\begin{corollary}
The familiar determinant relations from the commutative case also hold for
the noncommutative elementary and complete symmetric polynomials,%
\begin{equation}
h_{r}(\mathcal{A})=\det (e_{1-i+j}(\mathcal{A}))_{1\leq i,j\leq r},\qquad
e_{r}(\mathcal{A})=\det (h_{1-i+j}(\mathcal{A}))_{1\leq i,j\leq r}\;,
\label{ncJT}
\end{equation}%
where the determinants are well defined due to (\ref{ncehcomm}).
\end{corollary}

\begin{proof}
Performing a series expansion in (\ref{crystalTQ}) with respect to the
spectral parameter $u$ we find for $j=1,2,\ldots ,n$ the identities%
\begin{equation}
\sum_{r=0}^{j}(-1)^{r}e_{r}(\mathcal{A})h_{j-r}(\mathcal{A})=0~,
\end{equation}%
which constitute a system of homogeneous linear equations in a set of
commutative variables due to (\ref{ncehcomm}). The solution is therefore
identical to the commutative case and is given by (\ref{ncJT}); see \cite[%
page 21, eqn (2.6')]{MacDonald}. The second term on the right hand side of
equation (\ref{crystalTQ}) yields the equality%
\begin{equation*}
\sum_{r=0}^{n}(-1)^{r}e_{r}(\mathcal{A})h_{n+k-r}(\mathcal{A})\pi
_{k}=(-1)^{n}zh_{k}(\mathcal{A})\pi _{k}
\end{equation*}
which is easily verified by observing that $e_{n}(\mathcal{A})=z1$ and $%
h_{r}(\mathcal{A})\pi _{k}=0$ for $r>k$ as discussed earlier.
\end{proof}

\section{The WZNW fusion ring}

In this final section we explain how the ring of noncommutative functions
generated from the transfer matrix (\ref{Qmatrix}) for the vertex model (\ref%
{Boltzmannweight}) and the transfer matrix (\ref{Tmatrix}) for the phase
model is related to the fusion ring of $\widehat{\mathfrak{sl}}(n)_k$ WZNW conformal field theory. First we
need to introduce special elements in the ring of noncommutative functions.
We do this in a similar manner as in \cite[Cor 6.8 and Cor 7.2]{KS} and
therefore omit the proof.

\begin{proposition}[Cauchy-identity]
Let $\lambda $ be a partition and define the following noncommutative
analogue of Schur polynomials%
\begin{equation}
s_{\lambda }(\mathcal{A})=\det [h_{\lambda _{i}-i+j}(\mathcal{A})]_{1\leq i,j\leq n}\;.
\label{ncSchurdef}
\end{equation}%
In particular, we have $s_{(r)}(\mathcal{A})=h_{r}(\mathcal{A})$ and $%
s_{(1^{r})}(\mathcal{A})=e_{r}(\mathcal{A})$, where $(r)$ and $(1^{r})$ are
a horizontal and vertical strip of length $r$. Then we have the generalised
Cauchy identity%
\begin{equation}
Q(u_{1})\cdots Q(u_{l})=\sum_{\lambda }s_{\lambda }(u_{1},\ldots
,u_{l})s_{\lambda }(\mathcal{A})  \label{ncCauchy}
\end{equation}%
for all $l>0$. Here $s_{\lambda }(u_{1},\ldots ,u_{l})$ is the standard,
\emph{commutative} Schur polynomial.
\end{proposition}

Let us now recall the definition of the fusion ring ${\mathcal{F}}_{k}(%
\widehat{\mathfrak{sl}}(n),\mathbb{Z})$. The basic building blocks of the
WZNW conformal field theory are the primary fields which can be viewed as
the highest weight vectors with respect to the actions of both, the Virasoro
algebra and the affine algebra $\widehat{\mathfrak{sl}}(n)$. Hence, as a set
the primary fields are in one-to-one correspondence with certain elements in
the weight lattice of $\widehat{\mathfrak{sl}}(n)$ which we now describe.

We identify the basis vectors of $\mathcal{H}_{k}$ in (\ref{decomp})
labeled by $P_{k}^{+}$ (the finite set of partitions $\hat{\lambda}$ of
maximal height $n$ and of width $\hat{\lambda}_{1}=k$) with the {\em integral dominant
weights} of the affine algebra $\widehat{\mathfrak{sl}}(n)$ at level $k\in
\mathbb{Z}_{\geq 0}$. Namely, given a partition $\hat{\lambda}$ we map to
the weight $\sum_{i=1}^{n}m_{i}(\hat{\lambda})\hat{\omega}_{i}$, where the
coefficients $m_{i}(\hat{\lambda})$ are the multiplicities of columns of
height $i$ and the $\hat{\omega}_{i}$ are the fundamental $\widehat{%
\mathfrak{sl}}(n)$ weights with $\hat{\omega}_{n}\equiv \hat{\omega}_{0}$;
for details the reader is referred to \cite{Kac} and \cite[Part I, Section 2]{KS}. By abuse of notation we shall not distinguish between
partitions and weights.

Given two primary fields associated with two $\widehat{\mathfrak{sl}}(n)$
weights at level $k$, say $\hat{\lambda}$ and $\hat{\mu}$, their fusion
product can be expanded again into a sum of primary fields; see e.g. \cite%
{CFTbook}. Thus, we now consider the free abelian group ${\mathcal{F}}_{k}(%
\widehat{\mathfrak{sl}}(n),\mathbb{Z})$ generated by the elements in $%
P_{k}^{+}$ with respect to addition and introduce for $\hat{\lambda},\hat{\mu%
}\in P_{k}^{+}$ the fusion product as follows
\begin{equation}
\hat{\lambda}\ast \hat{\mu}=\sum_{\hat{\nu}\in P_{k}^{+}}\mathcal{N}_{\hat{%
\lambda}\hat{\mu}}^{(k),\hat{\nu}}\hat{\nu},\qquad \mathcal{N}_{\hat{\lambda}%
\hat{\mu}}^{(k),\hat{\nu}}=\sum_{\hat{\sigma}\in P_{k}^{+}}\frac{\mathcal{S}%
_{\hat{\lambda}\hat{\sigma}}\mathcal{S}_{\hat{\mu}\hat{\sigma}}\mathcal{\bar{%
S}}_{\hat{\nu}\hat{\sigma}}}{\mathcal{S}_{0\hat{\sigma}}}\;.  \label{OPE}
\end{equation}%
The structure constants $\mathcal{N}_{\hat{\lambda}\hat{\mu}}^{(k),\hat{\nu}%
} $, known as fusion coefficients in the physics literature, are given in
terms of the Verlinde formula \cite{Verlinde} with $\mathcal{S}_{\hat{\lambda%
}\hat{\mu}}$ denoting a matrix element of the modular $\mathcal{S}$-matrix
which is explicitly given in terms of the Kac-Peterson formula \cite%
{KacPeterson} ($\iota =\sqrt{-1}$),
\begin{equation}
\mathcal{S}_{\hat{\lambda}\hat{\sigma}}=\frac{e^{\iota \pi n(n-1)/4}}{\sqrt{%
n(k+n)^{n-1}}}\sum_{w\in S_{n}}(-1)^{\ell (w)}e^{-\frac{2\pi \iota }{k+n}%
(\sigma +\rho ,w(\lambda +\rho ))}  \label{modS}
\end{equation}%
Here $\lambda =\hat{\lambda}-k\hat{\omega}_{0}$ is the finite part of the
affine weight and $\rho =\sum_{i=1}^{n-1}\omega _{i}\;$is the Weyl vector
with $\omega _{i}$ being the finite fundamental weights of $\mathfrak{sl}(n)$%
.

\begin{theorem}[Korff, Stroppel \protect\cite{KS}]
Let $\hat{\lambda},\hat{\mu},\hat{\nu}\in P_{k}^{+}$ and set $z=1$. Then we
have the identity%
\begin{equation}
\hat{\lambda}\ast \hat{\mu}=s_{\hat{\lambda}}(\mathcal{A})\hat{\mu}
\label{combproduct}
\end{equation}%
and in particular the following equality holds, $\mathcal{N}_{\hat{\lambda}%
\hat{\mu}}^{(k),\hat{\nu}}=\left\langle \hat{\nu},s_{\hat{\lambda}}(\mathcal{%
A})\hat{\mu}\right\rangle $.
\end{theorem}

\begin{remark}
\textrm{Setting alternatively $z=0$ we specialise to the ring of
noncommutative Schur polynomials $s_{\lambda }(\mathcal{A}^{\prime })$ in
the local finite plactic algebra \cite{FG}. According to \cite[Lemma 6.3 and
Theorem 6.20]{KS} one then obtains the following quotient of the cohomology
ring of the Grassmannian $\func{Gr}_{k,n+k-1},$ $\mathrm{H^{\ast }(\func{Gr}%
_{k,n+k-1})}/\langle h_{n+k}\rangle \cong \mathbb{Z}[e_{1},\ldots
,e_{k}]/\langle h_{n},\ldots ,h_{n+k}\rangle $ , whose structure constants
are the intersection numbers }$c_{\lambda \mu }^{\nu }$\textrm{\ of three
hyperplanes with }$\mu _{1}=\nu _{1}$ \textrm{and coincide with the
celebrated Littlewood-Richardson coefficients \cite[\S 9.4, Exercise 21 (a)]%
{FultonYT}. }
\end{remark}

\begin{proof}
The proof of this result can be found in detail in \cite[Part I, Section 6]%
{KS} and relies on the explicit construction of an eigenbasis for the
transfer matrix (\ref{Tmatrix}) of the phase model (i.e. the generating
function of the noncommutative elementary symmetric polynomials (\ref{ncedef}%
)) via the quantum inverse scattering method or algebraic Bethe Ansatz; see
e.g. \cite{KBI}. Because of the relations (\ref{ncJT}) and (\ref{ncSchurdef}%
) this eigenbasis, called Bethe vectors, forms also an eigenbasis of the
other noncommutative symmetric polynomials and in particular of the transfer
matrix (\ref{Qmatrix}). One then verifies that the transformation matrix
from the standard basis labeled by $\hat{\lambda}\in P_{k}^{+}$ to the
Bethe vectors is given by the modular S-matrix (\ref{modS}). From this
result one derives the Verlinde formula (\ref{OPE}) for the matrix elements
of the noncommutative Schur polynomial and, thus, the identity (\ref%
{combproduct}) follows.
\end{proof}

We conclude by stating two corollaries which are now obvious consequences of
the last Theorem. We therefore omit their proofs. The first one uses the
recursion formula (\ref{nchrecursion}) for noncommutative complete symmetric
polynomials to relate fusion coefficients at different level $k$; this is in
analogy with the recursion relation in \cite[Cor 7.4]{KS}.

\begin{figure}[tbp]
\begin{equation*}
\includegraphics[scale=0.28]{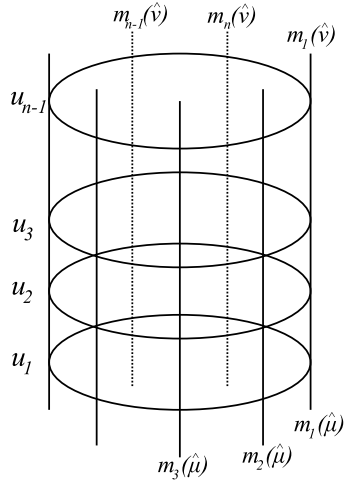}
\end{equation*}%
\caption{Graphical depiction of the $(n-1)\times n$ lattice with periodic
boundary conditions in the horizontal direction and fixed boundary
conditions $\hat{\protect\mu},\hat{\protect\nu}\in P_{k}^{+}$ on the outer
vertical edges. The spectral parameter varies from row to row. The
corresponding partition function obtained by summing over the Boltzmann
weights (\protect\ref{Boltzmannweight}) at each vertex yields (\protect\ref%
{Z}).}
\label{fig:partitionfunc}
\end{figure}

\begin{corollary}[Recursion relation]
For fixed level $k\in \mathbb{Z}_{\geq 0}$ let $\lambda =(r)$ be a
horizontal strip of length $r\leq k$ and set $\hat{\lambda}%
_{r}=(k,k-r,\ldots ,k-r)\in P_{k}^{+}$. Then we have the following recursion
relation for fusion coefficients,%
\begin{equation}
\mathcal{N}_{\hat{\lambda}_{r},\hat{\mu}}^{(k),\hat{\nu}}=c_{_{\hat{\mu}%
(r)}}^{\hat{\nu}}+\mathcal{N}_{\hat{\lambda}_{r-1},\varphi _{n}\hat{\mu}%
}^{(k-1),\varphi _{1}\hat{\nu}}\;,  \label{fusionrecursion}
\end{equation}%
where $c_{\lambda \mu }^{\nu }=c_{\mu \lambda }^{\nu }$ is the
Littlewood-Richardson coefficient.
\end{corollary}

Note that while this relation only involves horizontal strips the latter
allow one to compute all fusion coefficients via (\ref{ncSchurdef}) and (\ref%
{combproduct}).

The second consequence of the above Theorem and the identity (\ref{ncCauchy}%
) is the interpretation of the partition function of the vertex model (\ref%
{Boltzmannweight}).

\begin{corollary}[Generating function for fusion coefficients]
Given $\hat{\mu},\hat{\nu}\in P_{k}^{+}$ consider the vertex model (\ref%
{Boltzmannweight}) on an $(n-1)\times n$ lattice with periodic boundary
conditions in the horizontal direction and fix the boundary conditions in
the vertical directions by $\hat{\mu}$ and $\hat{\nu}$; see Figure \ref%
{fig:partitionfunc}. Assign to each row the spectral parameter $u_{i}$, then
the corresponding partition function (i.e. the weighted sum over all allowed vertex
configurations) has the expansion%
\begin{eqnarray}
Z_{\hat{\mu}}^{\hat{\nu}}(u_{1},\ldots ,u_{n-1})&=&\langle m(\hat\nu)|Q(u_1)\cdots Q(u_{n-1})|m(\hat\mu)\rangle\nonumber\\
&=&\sum_{\hat{\lambda}\in
P_{k}^{+}}\mathcal{N}_{\hat{\lambda}\hat{\mu}}^{(k),\hat{\nu}%
}s_{\lambda}(u_{1},\ldots ,u_{n-1}),  \label{Z}
\end{eqnarray}%
where $\lambda$ is the partition obtained from $\hat\lambda$ by deleting all columns of height $n$. Therefore we might interpret $Z$ as generating function for the fusion
coefficients.
\end{corollary}

\bigskip \noindent \textbf{Acknowledgment.} The author would like to thank
Catharina Stroppel for many helpful discussions and comments on a draft
version of this paper and Masato Okado for providing him with reference \cite{Okado}. He also gratefully acknowledges the financial support
of a University Research Fellowship of the Royal Society. \bigskip
\appendix
\noindent{\Large \textbf{Appendix}}

\section{Derivation of the vertex model from $U_{q}\widehat{\mathfrak{sl}}(2)
$}

In this appendix we describe how the matrices (\ref{crystalR}) and (\ref%
{crystalS}) are obtained as a special limit from the $U_{q}\widehat{%
\mathfrak{sl}}(2)$-intertwiner $R(u,v;\mu ,\nu ):M_{\mu }(u)\otimes M_{\nu
}(v)\rightarrow M_{\mu }(u)\otimes M_{\nu }(v)$ satisfying the relation%
\begin{equation}
R\Delta (X)=\Delta ^{\text{op}}(X)R,\qquad X\in U_{q}\widehat{\mathfrak{sl}}%
(2)\text{\ ,}  \label{intertwining}
\end{equation}%
where $\Delta ^{\text{op}}$ is the coproduct (\ref{cop}) with the order of
the factors interchanged. Given that the coproducts $\Delta $ and $\Delta ^{%
\text{op}}$ are algebra homomorphisms it is sufficient to solve the
intertwining relation for $X=E_{i},F_{i},K_{i}^{\pm 1}$. These identities
provide us with a set of equations for the matrix elements of $R=R(u,v;\mu
,\nu )$ which enable us to compute them recursively. For convenience we make
the parameter transformations $\mu \rightarrow q/s$ and $v\rightarrow s^{-1}$%
. Setting as before $R(v_{a}\otimes v_{b})=\sum_{c,d\geq
0}R_{c,d}^{a,b}~v_{c}\otimes v_{d}$ one obtains for $X=K_{i}^{\pm 1}$ the
constraint
\begin{equation}
R_{c,d}^{a,b}=0\text{\quad unless\quad }a+b=c+d,  \label{charge}
\end{equation}%
and for $X=E_{0},F_{1}$ the recursion relations%
\begin{eqnarray}
R_{c,d}^{a+1,b} &=&\frac{(q^{2d+1}-us^{2}\nu q^{2a})R_{c-1,d}^{a,b}+us\nu
(1-s^{2}q^{2(a+c)})R_{c,d-1}^{a,b}}{\nu -us^{2}q^{2(a+b)+1}}\ ,  \label{rec1}
\\
R_{c,d}^{a,b+1} &=&\frac{(\nu -\nu ^{-1}q^{2(d+b+1)})R_{c-1,d}^{a,b}+s(\nu
q^{2c}-uq^{2b+1})R_{c,d-1}^{a,b}}{\nu -us^{2}q^{2(a+b)+1}}\ ,  \label{rec2}
\end{eqnarray}%
where matrix elements with negative indices are understood to be zero. Note
that as long as $u,\mu ,\nu $ are generic these homogeneous relations
determine $R$ up to a scalar factor. We choose the convention $%
R_{0,0}^{0,0}=1,$ then the above relations allow us to successively compute
all other matrix elements. It follows from the general axioms of a
quasi-triangular Hopf algebra (see e.g. \cite[Section 4.2, Prop 4.2.7]{CP})
that the result yields a solution to the Yang-Baxter equation.

Since $R_{0,0}^{0,0}=1$ and all the coefficients in the recursion relations
are regular at $q=0$ we can conclude that $R_{c,d}^{a,b}\in \mathbb{A}$. Let
$\mathbb{J\subset A}$ be the ideal generated by $q$ and denote by $\tilde{R}%
_{c,d}^{a,b}$ the image of $R_{c,d}^{a,b}$ under the isomorphism $\mathbb{A}/%
\mathbb{J}\rightarrow \mathbb{C}$ then
\begin{eqnarray}
\tilde{R}_{c,d}^{a+1,b} &=&-us^{2}\delta _{a,0}\tilde{R}%
_{c-1,d}^{0,b}+us(1-s^{2}\delta _{a,0}\delta _{c,0})\tilde{R}_{0,d-1}^{0,b}\
,  \label{crystalrec1} \\
\tilde{R}_{c,d}^{a,b+1} &=&\tilde{R}_{c-1,d}^{a,b}+s\delta _{c,0}\tilde{R}%
_{0,d-1}^{a,b}\ ,  \label{crystalrec2}
\end{eqnarray}%
Note that these relations are now independent of the parameter $\nu \in
\mathbb{C}$. The solution to these equations can be explicitly written down,
the non-vanishing matrix elements are%
\begin{gather}
\tilde{R}_{c,b-c}^{0,b}=s^{b-c},\qquad \tilde{R}%
_{b+1,a-1}^{a,b}=-u^{a}s^{a+1},  \label{sol1} \\
\tilde{R}_{b-\varepsilon ,a+\varepsilon }^{a,b}=u^{a}s^{a+\varepsilon
}(1-s^{2})\text{\quad for\quad }a>0,\;0\leq \varepsilon \leq b\;.
\label{sol2}
\end{gather}%
To simplify the result further we now take the limit $\mathcal{R}%
_{c,d}^{a,b}(u)=\lim_{s\rightarrow 0}s^{-d}\tilde{R}_{c,d}^{a,b}$ and obtain
(\ref{Boltzmannweight}).

The derivation of the operator (\ref{crystalS}) follows along similar lines.
The $U_{q}\widehat{\mathfrak{sl}}(2)$-intertwiner must satisfy the following
identity in $\limfunc{End}[M_{\mu }(u)\otimes M_{\mu }(v)\otimes M_{\nu
}(1)] $,%
\begin{equation*}
R_{12}(u/v;\mu ,\mu )R_{13}(u;\mu ,\nu )R_{23}(v;\mu ,\nu )=R_{23}(v;\mu
,\nu )R_{13}(u;\mu ,\nu )R_{12}(u/v;\mu ,\mu )\;.
\end{equation*}%
Once more, the Yang-Baxter equation is a direct consequence of one of the
axioms of a quasi-triangular Hopf algebra such as $U_{q}\widehat{\mathfrak{sl%
}}(2)$. Notice the different dependence on the parameters $\mu $ and $\nu $
in the equation. Similar as before we can consider the value $\tilde{R}$ of
the matrix elements $R_{c,d}^{a,b}(us/v;q/s,q/s)$ at $q=0$ and then take the
limit $\mathcal{S}(u)_{c,d}^{a,b}:=\lim_{s\rightarrow 0}$ $s^{c-b}\tilde{R}%
_{c,d}^{a,b}$ to find the operator (\ref{crystalS}).

It turns out that also the phase model is obtained from the $U_{q}\widehat{%
\mathfrak{sl}}(2)$-intertwiner for $M_{\mu }(u)\otimes M_{\nu }(v)$ in the
crystal limit, albeit choosing a different specialisation for $\mu $.
Namely, set $\mu =q^{2}$ then the crystal limit of $(1\otimes
q^{N})R(-uq;q^{2},\nu )$ is the operator (\ref{phaseL}) over $\mathbb{C}%
(u)^{2}\otimes \mathcal{M}\subset M_{\mu =q^{2}}(u)\otimes \mathcal{M}$,
where we identify the two-dimensional subspace in $M_{\mu =q^{2}}(u)$
spanned by $\{v_{0},v_{1}\}$ with $\mathbb{C}(u)^{2}$. This reduction to a
subspace is justified by observing that for $\mu =q^{2}$ it can be mapped
onto the fundamental evaluation module of $U_{q}\widehat{\mathfrak{sl}}(2)$;
see Remark \ref{reduction}.

Finally, also the operator $L^{\prime }$ is another special crystal limit of
the $U_{q}\widehat{\mathfrak{sl}}(2)$-intertwiner determined by (\ref{charge}%
) and the recursion relations (\ref{rec1}) and (\ref{rec2}). We summarise
the various relations in the following table (as before we set $s\rightarrow
0$ after taking the crystal limit):\medskip

\begin{center}
\begin{tabular}{|c|c|c|c|}
\hline\hline
$U_{q}\widehat{\mathfrak{sl}}(2)$-intertwiner$^{}$ & crystal limit & $\mu $
& $\nu $ \\ \hline\hline
\multicolumn{1}{|c|}{$(1\otimes s^{-N})R(us)$} & \multicolumn{1}{|c|}{$%
\mathcal{R}(u)$} & \multicolumn{1}{|c|}{$q/s$} & \multicolumn{1}{|c|}{
arbitrary} \\ \hline
\multicolumn{1}{|c|}{$(s^{N}\otimes 1)R(us)(1\otimes s^{-N})$} &
\multicolumn{1}{|c|}{$\mathcal{S}(u)$} & \multicolumn{1}{|c|}{$q/s$} &
\multicolumn{1}{|c|}{$q/s$} \\ \hline
\multicolumn{1}{|c|}{$(1\otimes q^{N})R(-uq)$} & \multicolumn{1}{|c|}{$L(u)$}
& \multicolumn{1}{|c|}{$q^{2}$} & \multicolumn{1}{|c|}{arbitrary} \\ \hline
\multicolumn{1}{|c|}{$(1\otimes q^{N})R(-uq/s)(s^{N}\otimes 1)$} &
\multicolumn{1}{|c|}{$L^{\prime }(u)$} & \multicolumn{1}{|c|}{$q^{2}$} &
\multicolumn{1}{|c|}{$q/s$} \\ \hline
\end{tabular}%
~.
\end{center}

\end{document}